\documentclass[conference]{IEEEtran}
\IEEEoverridecommandlockouts
\usepackage{multirow}
\usepackage{soul}
\usepackage{mathtools}
\usepackage{amsmath,amsfonts, amsthm}
\usepackage{algorithmic}
\usepackage{graphicx}
\usepackage{float}
\usepackage{subfig}
\usepackage{textcomp}
\usepackage{xcolor}
\usepackage{placeins}
\usepackage{hhline}
\usepackage{tabularx}
\usepackage{enumitem}
\usepackage{setspace}
\usepackage{booktabs}
\usepackage{adjustbox}
\usepackage[normalem]{ulem}
\usepackage{indentfirst}
\usepackage{titlesec}

\makeatletter
\newcommand{\mathleft}{\@fleqntrue\@mathmargin0pt}
\newcommand{\mathcenter}{\@fleqnfalse}
\makeatother

\newtheorem{problem}{\textbf{Problem}}

\newtheorem{thm}{Theorem}

\newtheorem{challenge}{Challenge}

\def\BibTeX{{\rm B\kern-.05em{\sc i\kern-.025em b}\kern-.08em
    T\kern-.1667em\lower.7ex\hbox{E}\kern-.125emX}}

\begin{document}
\title{Deep Generation of Heterogeneous Networks}

\author{\IEEEauthorblockN{
			Chen Ling,
			Carl Yang, and
			Liang Zhao\IEEEauthorrefmark{1}}
		\IEEEauthorblockA{Department of Computer Science, Emory University, USA\\}
		\IEEEauthorblockA{\{chen.ling, j.carlyang, liang.zhao\}@emory.edu}
		\IEEEauthorrefmark{1}Corresponding Author
	}
\maketitle

\begin{abstract}
Heterogeneous graphs are ubiquitous data structures that can inherently capture multi-type and multi-modal interactions between objects. In recent years, research on encoding heterogeneous graph into latent representations have enjoyed a rapid increase. However, its reverse process, namely how to construct heterogeneous graphs from underlying representations and distributions have not been well explored due to several challenges in 1) modeling the local heterogeneous semantic distribution; 2) preserving the graph-structured distributions over the local semantics; and 3) characterizing the global heterogeneous graph distributions. To address these challenges, we propose a novel framework for heterogeneous graph generation (HGEN) that jointly captures the semantic, structural, and global distributions of heterogeneous graphs. Specifically, we propose a heterogeneous walk generator that hierarchically generates meta-paths and their path instances. In addition, a novel heterogeneous graph assembler is developed that can sample and combine the generated meta-path instances (e.g., walks) into heterogeneous graphs in a stratified manner. Theoretical analysis on the preservation of heterogeneous graph patterns by the proposed generation process has been performed. Extensive experiments\footnote{https://github.com/lingchen0331/HGEN} on multiple real-world and synthetic heterogeneous graph datasets demonstrate the effectiveness of the proposed HGEN in generating realistic heterogeneous graphs.
\end{abstract}


\section{Introduction}
	As a ubiquitous data structure, the graph can model connections (i.e., edges) between individual objects (i.e., nodes). Tremendous efforts have been made to study various types of graph problems, resulting in a rich literature of related papers and methods \cite{sun2019graph, you2018graphrnn, yun2019graph, GNNBook2021, zhao2021event}. The study of graphs can be mainly categorized into two categories: 1) graph representation learning, which aims at encoding graph topological and semantic information into vector space \cite{wu2020comprehensive}; and 2) graph generation, which reversely aims at constructing graph-structured data from low-dimensional space containing the graph generation rules or distribution \cite{guo2020systematic}. In the past years, previous studies of graphs have been made mostly on homogeneous graphs, which are the graphs consist of nodes under the same type. However, as a generalization of the homogeneous graph, heterogeneous graphs are the graphs with multiple types of nodes which further result in multiple types of edges, such as citation networks~\cite{zhou2007co} and social networks~\cite{dong2012link}. Fig. \ref{fig: 1}(b) shows a citation network with author, paper, venue, and term as nodes and ``authorship'', ``containment'' and ``publishment'' as edges. The local semantics based on certain combinations of node types and edge types reflect the key patterns of heterogeneous graphs~\cite{sun2011pathsim, sun2013meta}. Such local semantics are typically represented as \emph{meta-path}, a sequence of node types and edge types. Meta-paths characterize the rich and diverse relations among nodes~\cite{sun2013meta, sun2013pathselclus}. For example, as shown in Fig. \ref{fig: 1}(b), two authors can be connected via a meta-path since they both contribute to a paper, while two authors can alternatively be connected because their papers are accepted at the same venue.

	\begin{figure*}[tbp]
 		\centerline{\includegraphics[width=0.9\textwidth]{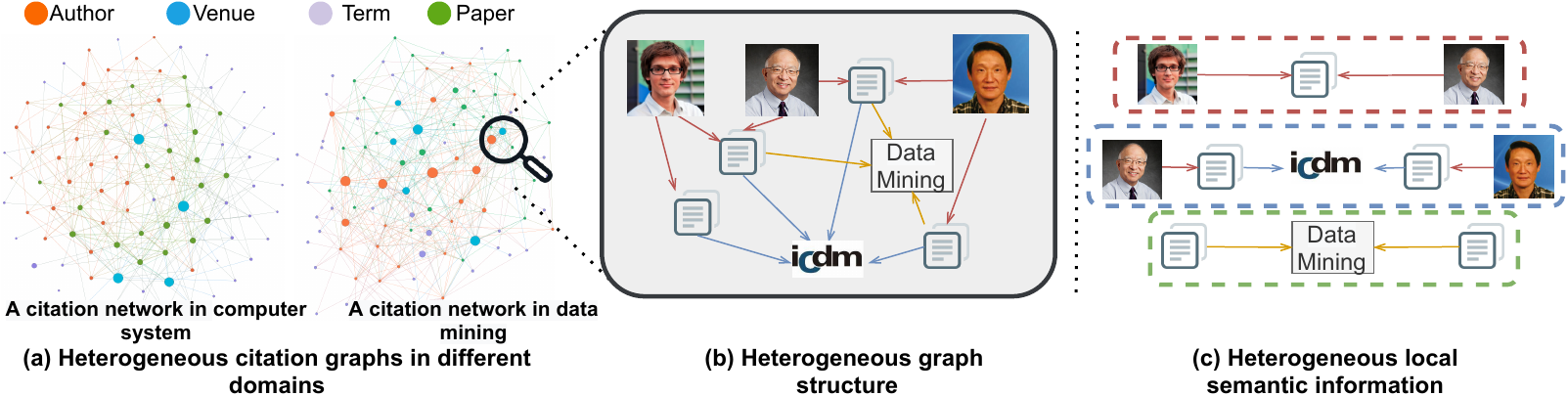}}
 		\caption{Examples of heterogeneous graph in academic field.}
 		\label{fig: 1}
 		\vspace{-5mm}
 	\end{figure*}
	
	
	Due to the recent advancement of various graph neural network models, plenty of works \cite{dong2017metapath2vec, fu2017hin2vec, wang2019heterogeneous, hu2020heterogeneous, yang2018similarity, yang2018meta, yang2019neural, yang2020heterogeneous, fan2020metagraph} have been proposed on studying heterogeneous graph representation learning and embedding in the past few years. These works have achieved significant progress in many downstream tasks (e.g., meta-relation detection \cite{dong2017metapath2vec, fu2017hin2vec}, heterogeneous node embedding learning \cite{wang2019heterogeneous, hu2020heterogeneous}, and heterogeneous link prediction \cite{zhang2019iteratively, yang2018similarity}). Among all the heterogeneous graph-related research, there remains a paucity of study on the \textit{heterogeneous graph generation}. It is self-evident that the advantages of generating realistic heterogeneous graphs are at least two-faceted: 1) generating high-quality heterogeneous graphs requires us to comprehensively capture the latent graph distribution, which can greatly enrich our understanding of the implicit properties of heterogeneous graphs; 2) generating heterogeneous graphs is useful in specific downstream applications (e.g., recommendation system \cite{shi2015semantic}, knowledge graph reasoning \cite{zhang2019iteratively}, and node proximity search \cite{sun2011pathsim}). Despite the importance of the heterogeneous graph generation, in the past decade, only one study \cite{gupta2012generating} tries to generate random heterogeneous graphs, which is based on hand-crafted rules and fails to generate realistic heterogeneous graph as it cannot learn the real data distribution underlying the observed graphs. On the other hand, a surge of research efforts on deep generative models~\cite{guo2020systematic, guo2019deep, bojchevski2018netgan, you2018graphrnn, simonovsky2018graphvae, yang2021secure} have been recently observed in the task of homogeneous graph generation. Through learning latent and complex dependencies directly from observed graphs, these deep graph generative models leverage different ways to learn and capture the underlying graph-structure distributions directly from the observed data without the need for hand-crafted rules. These approaches have been shown superiority in maintaining the structural properties in homogeneous graphs.

	However, existing deep generative models designed for homogeneous graphs cannot be trivially adapted to heterogeneous graphs due to the following significant challenges: 1) \textit{Difficulties in preserving heterogeneous semantic information.} Current works for homogeneous graphs have been either using random walks as a tool to learn the graph topological distribution as learning the distribution of random walks (\cite{bojchevski2018netgan, cannetgan}) or directly modeling an overall distribution of the edges (\cite{kipf2016variational, simonovsky2018graphvae}) over the homogeneous graphs. However, objects in heterogeneous graphs are inter-connected via various meta-paths as shown in Fig. \ref{fig: 1}(c). As the complex local semantic information is carried by meta-paths, adapting current works to the heterogeneous graph scenario without any elaborations on meta-path would bring difficulties in learning and preserving the distribution of such complex semantic patterns spanning different graph entities (i.e., edges and nodes) in the newly generated heterogeneous graphs.
	2) \textit{Difficulties in preserving heterogeneous higher-order structural information.} In some cases, meta-paths may also fall short of expressing more intricate relationships among nodes in heterogeneous graphs. As marked in Fig. \ref{fig: 1}(b), some common and symmetric higher-order structures spanning meta-paths will likely be observed repeatedly, which forms a triangle or orbit structure (e.g., one author writes two papers that are accepted by the same venue, and two papers of an author focus on the same research topic). The distributions of these higher-order graph structures are also hard to capture in heterogeneous graphs, which brings more challenges to effective heterogeneous graph generation. 3) \textit{Difficulties in preserving heterogeneous global information.} Meta-paths are also well-recognized to play a fundamental role in preserving the global patterns of heterogeneous graphs \cite{sun2011pathsim}. For example, the ratio of different node types, and edge types, and their meta-paths are apparently different between the citation networks of computer system domain and data mining domain, as shown in Fig. \ref{fig: 1}(a). It is important to preserve the global distribution of meta-path patterns during heterogeneous graph generation, which is again extremely difficult as it is entangled with the preservation of node type ratios, edge type ratios, and graph topological patterns.

    In coping with these challenges, we introduce an end-to-end graph generative framework, namely \underline{H}eterogeneous Graph \underline{Gen}eration (HGEN), whose goal is to generate novel heterogeneous graphs by preserving all the complex local semantic,  higher-order structural, and global properties through directly modeling the distribution of meta-paths in observed heterogeneous graphs. Particularly, to deal with the first challenge of capturing the complicated local semantics, we propose to learn a joint distribution of the random walks and the associated meta-paths from the observed heterogeneous graphs. On top of that, for challenge two, we encode heterogeneous higher-order structural information into nodes via embedding learning and use it to guide the generation of meta-paths and random walks that form different high-order heterogeneous structures. To tackle the third challenge, we develop a novel heterogeneous graph assembly method, which is theoretically proved to preserve the global heterogeneous graph patterns in node types, edge types, and meta-paths.
	We conclude our major contributions as follows:
	\begin{itemize}[leftmargin=*]
	    \item \textbf{Problem.} We not only formulate a new paradigm of heterogeneous graph generation but also identify and resolve its unique challenges in preserving various heterogeneous graph properties.
	    \item \textbf{Framework.} We propose an end-to-end generative framework for heterogeneous graph generation. The proposed framework can effectively learn the underlying distribution of heterogeneous graphs. It generates heterogeneous graphs with ensuring the preservation of various heterogeneous graph properties.
	    \item \textbf{Evaluation.} We conduct extensive experiments on both synthetic and real-world heterogeneous graphs. Compared with state-of-the-art baselines, HGEN achieves competitive results in preserving most of the static graph properties. In addition, HGEN is shown to be capable of generating realistic heterogeneous graphs by preserving important meta-path information.
	\end{itemize}
    
    
    \section{Related Work} \label{sec: related_works}
    \textbf{Heterogeneous Graph Mining.} Compared to the commonly-adopted homogeneous graph, heterogeneous graph carries much richer semantic information and has therefore gained much attention in recent literature \cite{shi2016survey}. The concept of meta-paths in heterogeneous graph \cite{sun2013pathselclus, sun2011pathsim} is one of the most important concepts proposed to capture numerous semantic relationships across multiple types of objects systematically. Since the introduction of heterogeneous graph, many innovative data mining tasks have spawned, including similarity search \cite{sun2011pathsim}, object clustering \cite{sun2013pathselclus}, and heterogeneous node classification \cite{wang2019heterogeneous}. 
    
    \textbf{Heterogeneous Graph Representation Learning.} In recent years, graph neural network (GNN) has achieved massive success in extensive applications \cite{kipf2016semi, yun2019graph} due to its capability of effectively learning relationships and interactions on non-Euclidean data. There exist plenty of attempts trying to adopt GNNs to learn with heterogeneous graphs, and almost all of them rely on employing meta-paths to model heterogeneous structures \cite{yang2020heterogeneous}. Specifically, proximity-preserving methods \cite{dong2017metapath2vec, fu2017hin2vec, yang2018similarity, yang2019neural} aim to capture heterogeneous network topological information via meta-path-constrained random walks. On the other line of approach, \cite{yang2018meta, wang2019heterogeneous, hu2020heterogeneous} try to aggregate information from heterogeneous neighbors via multiple layers of learnable projection functions. Throughout the study of heterogeneous graphs \cite{sun2011pathsim, yang2020heterogeneous}, \textit{meta-path} serves as the fundamental building block owing to its nonpareil ability to carry both graph topological and rich semantic information.
    
    \textbf{Graph Generation.} Generative models for graphs have a rich history due to the wide range of applications in different domains, such as link prediction \cite{bojchevski2018netgan, simonovsky2018graphvae}, protein structure analysis \cite{de2018molgan}, and information diffusion analysis in social networks \cite{wang2018graphgan}. Traditional graph generation methods (e.g., random graphs, stochastic block models, and Bayesian network models) fail to model complex dependencies in our real-world scenarios. In addition, they cannot effectively preserve the statistical properties of the observed graphs. In the last few years, there has been a surge in research focusing on deep graph generation. According to \cite{guo2020systematic}, the current deep graph generation can be divided into two categories: sequential-based and one-shot-based. For sequential-based graph generation methods \cite{you2018graphrnn, bojchevski2018netgan, sun2019graph}, they autoregressively generate the nodes and edges with the LSTM model. However, the sequential-based generation is limited in following a fixed node/edge permutation order, which greatly loses the generation flexibility and model scalability. On the other hand, one-shot-based generation methods \cite{simonovsky2018graphvae, de2018molgan, bojchevski2018netgan, yang2019conditional, yang2021secure, zhang2021tg,ling2022stgen} try to build a probabilistic graph model based on the matrix representation that can generate graph topology as well as node/edge attributes in a one-shot, but most of them cannot easily be applied in large graphs due to the large time complexity. Finally, multi-attributed graph generation \cite{you2018graphrnn, guo2019deep, goyal2020graphgen} aims at generating homogeneous graphs by preserving node/edge attributes. Instead, the key patterns of heterogeneous graphs are the higher-order local semantics reflected by the combinatorial of the types of nodes and edges, which cannot be captured by methods for homogeneous graphs.
    
    \section{Problem Formulation}\label{sec: prob}
 	A heterogeneous graph \cite{shi2016survey, yang2020heterogeneous} is a graph $\mathcal{G} = \{\mathcal{V}, \mathcal{E}\}$ with multiple types of objects and relations. $\mathcal{V}$ is the set of objects (i.e., nodes), where each node $v_i \in \mathcal{V}$ is associated with a node type $o = \phi(v_i)$. $\mathcal{E} \subseteq\mathcal{V}\times \mathcal{V}$ is the set of edges, where each edge $e_{ij} \in \mathcal{E}$ is associated with a relation type $l = \psi(e_{ij})$.
 	
    \begin{figure*}[!t]
    \centering
    \includegraphics[width=1.0\textwidth]{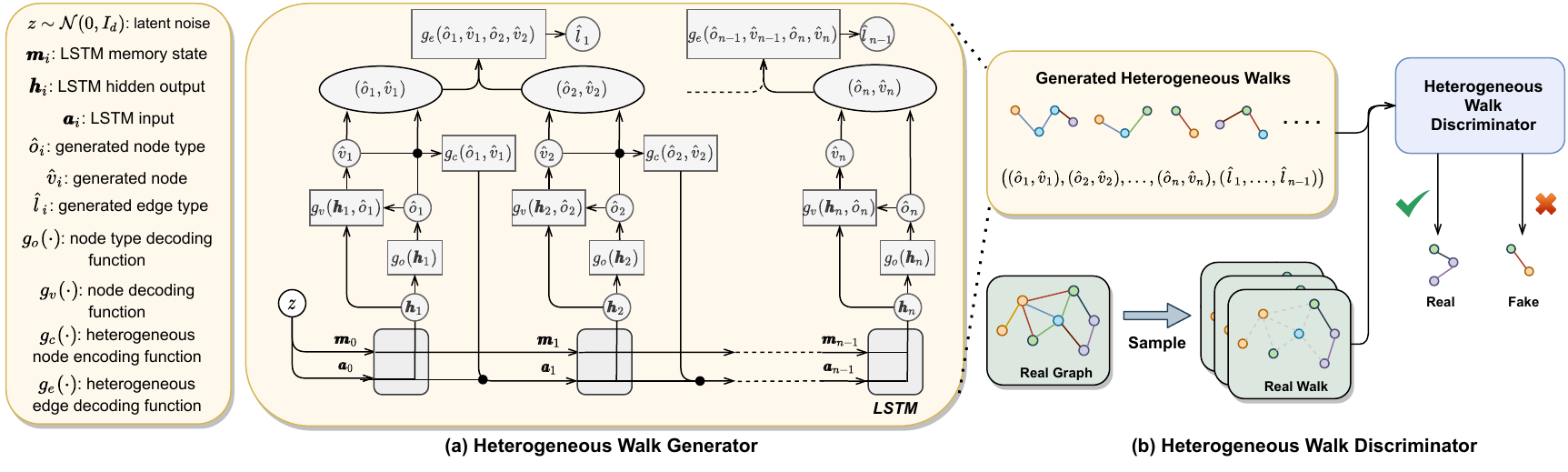}\vspace{-3mm}
    \caption{The illustration of the heterogeneous walks generation in HGEN.}
    \label{fig: model_framework}
    \vspace{-6mm}
    \end{figure*}
    
    In the study of heterogeneous graphs, the concepts of meta-paths are widely considered as cornerstones and adopted to systematically capture numerous semantic relationships across multiple types of objects, which are defined as a path over the graph \cite{sun2013pathselclus,yang2020heterogeneous}. Hence meta-paths are indispensable to be considered as basic units for heterogeneous graph generation. Concretely, a meta-path $\mathbf{o}$ is defined as a sequence of object types and edge types $\mathbf{o} =  \big((o_1, o_2, ..., o_n), (l_1, l_2, ..., l_{n-1})\big) = o_1 \xrightarrow{l_1} o_2  \xrightarrow{l_2} ...\xrightarrow{l_{n-1}} o_n$, where each $o_i$ and $l_j$ are node type and edge type in the sequence, respectively. Each meta-path captures the rich semantic information between its two ends $o_1$ and $o_n$. In heterogeneous graphs, the local semantic information is carried on each of walks $\mathbf{v} = (v_0, v_1, ..., v_n)$ and its associated meta-path $\mathbf{o}$. We again take Fig. \ref{fig: 1}(c) as an example, there exist two meta-paths between papers: (\textit{Paper, Author, Paper}) and (\textit{Paper, Venue, Paper}). The utilization of different meta-paths allow the heterogeneous graph to contain rich topological and semantics among diverse objects, which has been shown beneficial to many real-world graph mining applications \cite{yang2020heterogeneous, wang2019heterogeneous, hu2020heterogeneous}.
    
    
    With the preliminary notion of the heterogeneous graph, we formalize the heterogeneous graph generation problem as follows:
    \begin{problem}[\textbf{Heterogeneous Graph Generation}] 
    The goal of the heterogeneous graph generation is to learn a distribution $p_{\text{data}}(\mathcal{G})$ from the observed heterogeneous graphs such that a new graph $\hat{\mathcal{G}}$ can be obtained by sampling $\hat{\mathcal{G}} \sim p_{\text{data}}(\mathcal{G})$.  
    \end{problem}
    \begin{challenge}[\textbf{Difficulties in modeling the complex local semantic information.}]
    Although the existence of meta-paths allows heterogeneous graph to characterize the combinatorial of node types and edge types, it is unclear how to model their distributions and generatively assemble them into heterogeneous graphs.
    \end{challenge}
    
    \begin{challenge}[\textbf{Difficulties in characterizing the heterogeneous structural patterns.}]
    The local structural patterns in heterogeneous graphs are often expressed in higher-order proximity among the nodes and edges (e.g., triangles, orbits, and other higher-order structures). Such the local structure may fuse multiple walks under one or more meta-paths with richer semantic information, yet brings more difficulties in learning its distribution.
    \end{challenge}
    \begin{challenge}[\textbf{Difficulties in capturing heterogeneous global meta-path information.}]
    Meta-paths indeed play a significant role in preserving the global patterns of heterogeneous graphs. In heterogeneous graph generation, it is important yet challenging to preserve the global distribution of meta-path patterns since the distribution of meta-path patterns often involves node type ratios, edge type ratios, and graph topological patterns.
    \end{challenge}
    
    \section{Heterogeneous Graph Generation}\label{sec: model}
    To address the above challenges, we propose a new heterogeneous graph generation framework, named HGEN. To address the first and second challenge, we propose a \textit{heterogeneous walk generator} in Sec. \ref{sec: generator} to jointly learn the distribution of local walks and the associated meta-paths so that both heterogeneous topological and local semantic information can be well captured. To overcome the second challenge, we leverage the heterogeneous node embedding to make the generator be aware of any potential higher-order structures that each node may be involved with. Finally, for the third challenge, we propose a novel \textit{heterogeneous graph assembler} in Sec. \ref{sec: framework}, which can construct new heterogeneous graphs by capturing the global heterogeneous property, namely different meta-path ratios. We further prove that the global heterogeneous property can be well-preserved through our Theorem \ref{thm: 1} introduced in Sec. \ref{sec: proof}.
    
    \subsection{Heterogeneous Walk Generator} \label{sec: generator}
    In the observed graph $\mathcal{G}$, a heterogeneous walk is defined as a tuple that consists of two components: a walk $\mathbf{v}$ and an associated meta-path $\mathbf{o}$. The proposed heterogeneous walk generator $G$ is defined as a probabilistic sequential learning model to generate synthetic heterogeneous walks: $(\hat{\mathbf{v}}, \hat{\mathbf{o}}) = \big((\hat{v}_1, \hat{v}_2, ..., \hat{v}_n), ((\hat{o}_1, \hat{o}_2, ..., \hat{o}_n), (\hat{l}_1, \hat{l}_2, ..., \hat{l}_{n-1}))\big)$, where the $\hat{\mathbf{v}}$ and $\hat{\mathbf{o}}$ are denoted as the generated walk and associated meta-path, respectively. We use $\hat{v}_i$, $\hat{o}_i$, and $\hat{l}_i$ to denote each of the generated node, node type, and edge type in $(\hat{\mathbf{v}}, \hat{\mathbf{o}})$, respectively. Fig. \ref{fig: model_framework}(a) illustratively summarizes the whole generative process of each synthetic heterogeneous walk. 
    
    \textbf{Heterogeneous Walk Generation.} We model $G$ as a sequential learning process based on a recurrent architecture, and each unit $f_{\theta}$ in the sequential model is parameterized by $\theta$ so that it can generate a node type $\hat{o}$ and a corresponding node $\hat{v}$ that belongs to this node type in a hierarchical manner. Precisely, the node type $\hat{o}$ is determined based on the previously generated sequence, and the node $\hat{v}$ is then coherently determined by the generated node type as well as the generated sequence. Both generated node type $\hat{o}$ and node $\hat{v}$ together provide information for the generation of the next node type and node instance.
    
    Specifically, at each recurrent block (i.e., time step) $t$, $f_{\theta}$ produces two outputs $(\pmb{m}_t, \pmb{h}_t)$, where the $\pmb{m}_t$ is the current memory state and the $\pmb{h}_t$ is a latent probabilistic distribution (i.e., hidden output of $f_{\theta}$) denoting the information carried from previous time steps. We first sample the node type $\hat{o}_t \sim g_o(\pmb{h}_t)$ based on the probability distribution $\pmb{h}_t$, where the $g_o(\cdot)$ is a node type decoding function. We then sample the node $\hat{v}_t$ by a node decoding function $\hat{v}_t \sim g_v(\pmb{h}_t, \hat{o}_t)$ that takes $\pmb{h}_t$ and $\hat{o}_t$ as inputs. Lastly, the generated node type $\hat{o}_t$ and node $\pmb{h}_t$ are fused by a heterogeneous node encoding function $g_c(\hat{o}_t, \hat{v}_t)$, which then serves as the input of next recurrent block. 
    
    \textbf{Heterogeneous Node Sampling.} 
    To overcome the second challenge, we cannot uniformly sample $\hat{v}_t$ based on the node type $\hat{o}_t$ because such a way may cause the neglection of (1) \textit{node structural} distribution and (2) \textit{node semantic} distribution. For example, we may observe an author always tends to cite a paper with high citation (namely, high node degree of this paper node). Then such distribution needs to be modeled with structural information. On the other hand, we may observe a data mining paper is unlikely to cite a computer system paper, and we may also need to characterize this tendency in the distribution. Both of the above distributions cannot be tackled by uniformly sampling. Therefore, to tackle this challenge, since latent node embedding could encode both topological and semantic information into the node, we propose to calculate a latent embedding $\Tilde{v}_t$ of the next node $v_t$, then we select with a higher probability the closer embedding among all the embeddings that belong to node type $\hat{o}_t$ so that the next node $v_t$ can be determined by the sampled embedding. 

    More specifically, we first calculate the latent node embedding $\Tilde{v}_t$ based on the sampled node type $\hat{o}_t$ by a simple linear transformation. We then calculated the distance between $\Tilde{v}_t$ and other node embedding $\Tilde{v}^{(\hat{o}_t)}_i$, meaning any node $\Tilde{v}_i$ belonging to the sampled node type $\hat{o}_t$. In this case, given a total number of $k$ embeddings  that belong to the type $\hat{o}_t$, the next node $\hat{v}_t$ can be sampled from a multinomial distribution:
    \begin{equation*}
    \hat{v}_t \sim \text{Multi}(\Tilde{v}^{(\hat{o}_t)}_1, \Tilde{v}^{(\hat{o}_t)}_2, ..., \Tilde{v}^{(\hat{o}_t)}_k; p_1, p_2, ..., p_k),
    \end{equation*} where each $p_i = -\rVert{d(\Tilde{v}_t, \Tilde{v}^{(\hat{o}_t)}_{i})}\rVert^2$  and $d(\cdot, \cdot)$ is a distance metric such as Euclidean distance. Note that the node embedding $\Tilde{v}^{(\hat{o}_t)}_i$ can be obtained from a conventional heterogeneous node embedding technique such as \cite{fu2017hin2vec}. 
    
    In order to generate a variable-length heterogeneous walk, we incorporate a end-of-sequence token as an additional node type so that the heterogeneous walk generator stops when the sampled node type is the token at any steps. Therefore, the proposed generator is able to produce variable-length heterogeneous walks. Finally, the edge type $l_t$  can be predicted by a simple edge decoding function $g_e(\hat{o}_{t}, \hat{v}_{t}, \hat{o}_{t-1}, \hat{v}_{t-1})$ that takes its two end nodes $\hat{v}_{t-1}$ and $\hat{v}_{t}$ as well as their node types $\hat{o}_{t-1}$ and $\hat{o}_{t}$ as inputs.
    In all, we summarize the overall generative process as follows:
    
    \resizebox{.99\linewidth}{!}{
  \begin{minipage}{\linewidth}
  \begin{align*}
        &\pmb{a}_0 = 0,\; \pmb{m}_0 = f_0(\pmb{z}), \; \pmb{z}\sim \mathcal{N}(0, 1)\\
        &\pmb{a}_1 = g_c(\hat{o}_1, \hat{v}_1), \: \hat{v}_1 \sim g_v(\pmb{h}_1, \hat{o}_1),\: \hat{o}_1 \sim g_o(\pmb{h}_1), (\pmb{m}_1, \pmb{h}_1) = f_{\theta}(\pmb{m}_0, \pmb{a}_0)\\
        &\pmb{a}_2 = g_c(\hat{o}_2, \hat{v}_2), \: \hat{v}_2 \sim g_v(\pmb{h}_2, \hat{o}_2),\: \hat{o}_2 \sim g_o(\pmb{h}_2), (\pmb{m}_2, \pmb{h}_2) = f_{\theta}(\pmb{m}_1, \pmb{a}_1)\\
        &\hat{l}_1 = g_e(\hat{o}_2, \hat{v}_2, \hat{o}_1, \hat{v}_1)\\
        &\cdots\\
        &\hat{v}_n \sim g_v(\pmb{h}_n, \hat{o}_n),\: \hat{o}_n \sim g_o(\pmb{h}_n),\: (\pmb{m}_n, \pmb{h}_n) = f_{\theta}(\pmb{m}_{n-1}, \pmb{a}_{n-1})\\
        &\hat{l}_{n-1} = g_e(\hat{o}_{n}, \hat{v}_n, \hat{o}_{n-1}, \hat{v}_{n-1})\\
    \end{align*}
  \end{minipage}
}

    
    In this work, we utilize LSTM as the recurrent architecture, and $f_{\theta}$ becomes a single LSTM unit. To initialize the whole generative process, $G$ takes a random noise $\pmb{z}$ as input, which is drawn from a standard Gaussian distribution. Additionally, for the node type decoding function $g_o(\cdot)$, we apply the Gumbel-softmax trick \cite{jang2016categorical} in $g_o(\cdot)$ to make the whole sampling differentiable. Finally, in most of the real-world scenarios, the edge type $l_t$ can be determined by the types of its two end nodes $\hat{o}_t$ and $\hat{o}_{t-1}$ if there does not exist multi-typed relations between two node types. In this case, the heterogeneous walk generator can be simplified only to generate node sequences and associated node types. 
    
    \begin{figure*}[!t]
    \centering
    \includegraphics[width=0.85\textwidth]{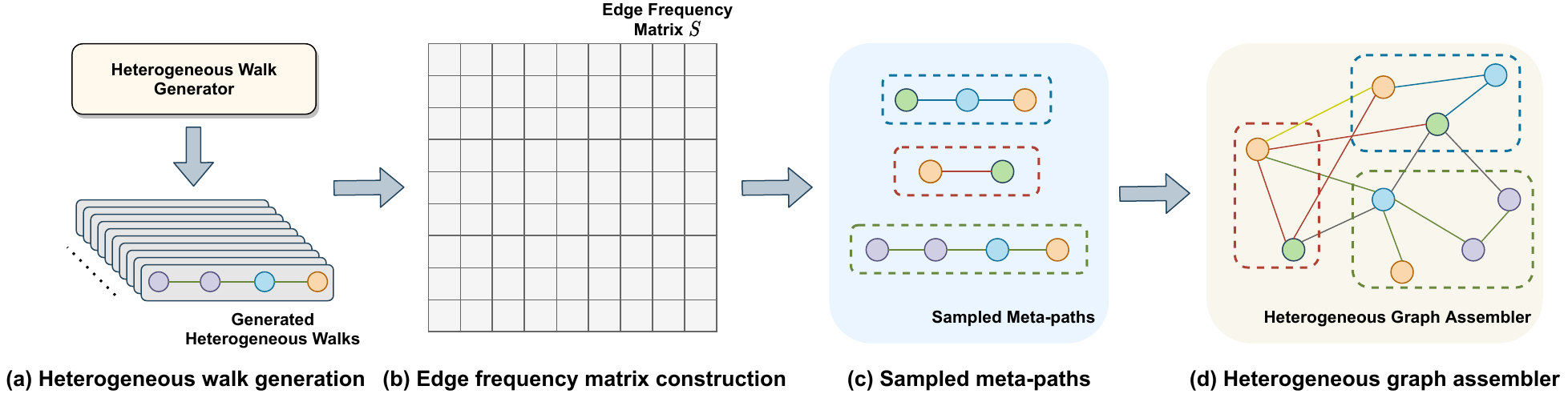}
    \vspace{-2mm}
    \caption{The process of heterogeneous graph assembler.}
    \label{fig: assembler}
    \vspace{-6mm}
    \end{figure*}

    \subsection{Heterogeneous Generator Training and Utilization} \label{sec: framework}
    In the following, we will introduce how to train the above-mentioned generator and how to use the heterogeneous walks generated by it to construct heterogeneous graphs. Concretely, we utilize a heterogeneous discriminator $D$ to distinguish between real and fake heterogeneous walks, where the real heterogeneous walks are uniformly sampled from the observed graph. We then propose a heterogeneous graph assembler to construct new graphs based on the sampled heterogeneous walks. More details are presented as follows.
    
    We first introduce the overall objective function of the Wasserstein heterogeneous GAN \cite{arjovsky2017wasserstein}, which is written as:
    \begin{equation} \label{eq: obj}
        \begin{split}
        \mathcal{L}_{\text{HGEN}} = &\text{ max}\:\mathbb{E}_{ (\mathbf{o}, \mathbf{v})\sim p(\mathcal{G})}[D_o(\mathbf{o})+ D_v(\mathbf{v})] \\
        &-\mathbb{E}_{ z\sim p(z)}[D_o(\hat{\mathbf{o}})+ D_v(\hat{\mathbf{v}})],
        s.t. \,\, G(z) = (\hat{\mathbf{o}}, \hat{\mathbf{v}}),
    \end{split}
    \end{equation}
    where $\mathbf{v}$ and $\mathbf{o}$ are the random walk and associated meta-path, respectively, directly sampled from the observed heterogeneous graph $\mathcal{G}$. They are the real data for training our heterogeneous walk generator $G$. Specifically, given an observed heterogeneous graph $\mathcal{G} = \{\mathcal{V}, \mathcal{E}\}$, we utilize random-walk-based method to uniformly sample a set of random walks $\{\textbf{v}_1, \textbf{v}_2, ...\}$, where each $\textbf{v}_i$ is a node sequence s.t. $\textbf{v}_i = (v_1, v_2, ..., v_n)$. In addition, we extract the meta-path information $\textbf{o}_i = \big((o_1, o_2, ..., o_n), (l_1, l_2, l_{n-1})\big)$ from each $\textbf{v}_i$. 
    
    The heterogeneous discriminator $D$ in Eq. \eqref{eq: obj} is designed as a parallel recurrent architecture in order to individually distinguish whether each component in the heterogeneous walks are valid or not. Specifically, at each recurrent block (i.e., each step) $t$, the discriminator $D$ takes two inputs: the generated node type $\hat{o}_t$ and node index $\hat{v}_{t}$, each of which is fed into an individual recurrent unit. After processing both sequences, the discriminator returns a single score $D_v(\mathbf{v})+D_o(\mathbf{o})$ that represents the probability of the heterogeneous walk being real.
    
    \textbf{Heterogeneous Graph Assembler.} To assemble a heterogeneous graph from the generated heterogeneous walks, we further propose a novel stratified heterogeneous edge sampling strategy to achieve the following steps: 1) it first samples a node $\hat{v}_i$ and its type $\hat{o}_i$ from all of the generated heterogeneous walks; 2) based on the node type $\hat{o}_i$, we then sample a meta-path that starts with $\hat{o}_i$; 3) we iteratively sample the next node $\hat{v}_{i+1}$ in the sampled meta-path if both of the node type $\hat{o}_{i+1}$ and edge type $\hat{l}_{i}$ fits the meta-path pattern. 
    
    More specifically, the generator $G$ firstly produces a sufficient number of heterogeneous walks as shown in Fig. \ref{fig: assembler}(a). We then construct an symmetric adjacency matrix $S$ with size $|\mathcal{V}|\times |\mathcal{V}|$ to record the count of edges observed from the sampled heterogeneous walks in each entry $S_{ij}$, where the $|\mathcal{V}|$ is the size of the node set. Next, we collect all of the meta-path patterns generated by the generated heterogeneous walks, as shown in Fig. \ref{fig: assembler}(b-c). For the first step of the stratified heterogeneous edge sampling, we sample the a node $\hat{v}_i$ and its type type $\hat{o}_i$ based on the node degree distribution $\frac{\sum_jS_{ij}}{|\mathcal{V}|}$. For the second step, among all the meta-paths $\{\mathbf{o}^{(f)}_1, \mathbf{o}^{(f)}_2, ...\}$ that start with the node type $\hat{o}_i$, we sample a meta-path $\mathbf{o}^{(f)}_i$ based on the probability $\frac{c(\mathbf{o}^{(f)}_i)}{T^{\hat{o}_i}}$, where $T^{\hat{o}_i}$ is the total count of generated meta-paths that starts with node type $\hat{o}_i$ and $c(\mathbf{o}^{(f)}_i)$ is the count of meta-path pattern $\mathbf{o}^{(f)}_i$. For the third step, by following this meta-path pattern $\mathbf{o}_r = (o_1, o_2, ..., o_n)$, we iteratively sample all the nodes whose node types are regulated by the the meta-path. Precisely, we sample the next node $v_j$ by sampling all the neighbors of the current node $v_i$ with the probability $p_{v_iv_j} = (S_{ij})/(\sum_{s} S_{is})$ such that all the nodes $v_s$ belong to the specific node type $o_j$ following the meta-path $\mathbf{o}_i^{(f)}$. The sampled node sequence $\mathbf{v}_r = (v_0, v_1, ...)$ is then added to the current under construction. We continue the stratified heterogeneous edge sampling strategy until the desired amount of edges is reached. The final assembled graph is visualized in Fig. \ref{fig: assembler} (d).
    
    \textbf{Complexity Analysis.} The computational complexity of HGEN is $O(W\cdot L)$, where $W$ is the weights of a single LSTM unit, and $L$ is the length of the generated heterogeneous walks. However, the length of our proposed heterogeneous walk is considerably small ($1\le L \le 3$) while the walk length in other random-walk-based graph generative method \cite{bojchevski2018netgan} is $\ge 16$. For auto-regressive graph generation models \cite{yun2019graph, you2018graphrnn}, the time complexities are at least $O(|\mathcal{V}|^2\cdot W)$, where $|\mathcal{V}|$ is the cardinality of the node set. They convert graph as a long sequence by performing a large number of breadth-first-search enumerations for each graph. Additionally, HGEN also has linear complexity in graph assembly, it only needs to run the trained model $T_s$ times to sample heterogeneous walks for constructing the score matrix $S$. To sum up, the overall complexity of HGEN can be reduced to $O(W + T_s)$, which makes our proposed model highly efficient for handling large graphs, since the overall process is not sensitive to $|\mathcal{V}|$ at all.
    
    \begin{figure}[!t]
 		\centerline{\includegraphics[width=0.4\textwidth]{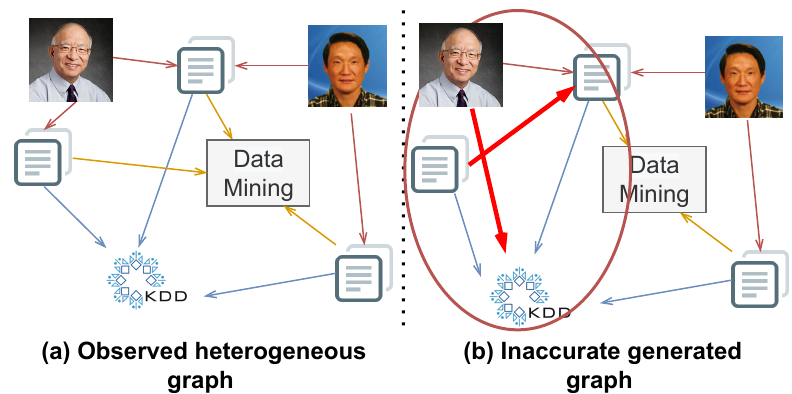}}
 		\caption{Example of two heterogeneous graphs with different semantic information: the observed meta-path patterns are different, although the node and edge distribution are the same between two graphs. Specifically, since we do not observe a direct link between (author, venue) and (paper, paper) in the observed graph Fig. \ref{fig: proof}(a). It is not accurate for the generated graph Fig. \ref{fig: proof}(b) that generate such links.}
 		\label{fig: proof}
 		\vspace{-6mm}
 	\end{figure}

    \subsection{Meta-path Information Preservation Analysis} \label{sec: proof}
    As we discussed in Sec. \ref{sec: prob}, it is significant to preserve the meta-path information in our generated graph. Taking Fig. \ref{fig: proof} as an example, although both graphs have exactly the same structure, they are still regarded as two different heterogeneous graphs since their meta-path distributions are different. Given the importance of the meta-path information in heterogeneous graph generation, we further show that our framework can successfully preserve this meta-path information as proved in Theorem \ref{thm: 1}.
    
    \begin{thm} \label{thm: 1}
        The distribution of meta-path patterns $\overline{\mathcal{O}}^{(r)}$ of the generated heterogeneous graph equals the distribution of meta-path patterns $\overline{\mathcal{O}}$ in the observed heterogeneous graph, namely $p(\overline{\mathcal{O}}^{(r)}) = p(\overline{\mathcal{O}})$.
    \end{thm} 
    \begin{proof}
        We will prove that the ratio of the meta-path patterns can be preserved in three steps: 1) the ratio of different meta-path patterns can be preserved during the sampling procedure; 2) the ratio of generated meta-path patterns can be preserved during the generation procedure; 3) the meta-path patterns can be preserved during the graph assembling procedure.
        
        \textit{Meta-path Ratio Preservation in Sampling.} Let $\overline{\mathcal{O}} = (\overline{\mathbf{o}}_1, \overline{\mathbf{o}}_2, ...)$ be the collection of meta-paths obtained from the observed heterogeneous graph $\mathcal{G}$, each $\overline{\mathbf{o}}_i$ is a meta-path in one-hot format $\overline{\mathbf{o}}_i \in \{0, 1\}^{1\times R}$, where the $R$ is the total number of different meta-path patterns. $\overline{\mathcal{O}}^{(\tau)} = (\overline{\mathbf{o}}^{(\tau)}_1, \overline{\mathbf{o}}^{(\tau)}_2, ..., \overline{\mathbf{o}}^{(\tau)}_K)$ is the sequence of sampled meta-paths with sampling size $K$, where each meta-path $\overline{\mathbf{o}}^{(\tau)}_j \in \{0, 1\}^{1\times R}$ is drawn independent and identically distributed ($i.i.d$) from $\overline{\mathcal{O}}$. 
        
        Suppose that $\boldsymbol{\mu} = [\mu_1, \mu_2, ..., \mu_R]^T$ denotes the probability of each individual meta-path pattern in $\overline{\mathcal{O}}$, it is obvious that $\mathbb{E}[\overline{\mathbf{o}}_i|\boldsymbol{\mu}] = \sum_{\overline{\mathbf{o}}_i}p(\overline{\mathbf{o}}_i|\boldsymbol{\mu})\overline{\mathbf{o}}_i = [\mu_1, \mu_2, ..., \mu_R]^T = \boldsymbol{\mu}$. Now consider the total $K$ observations $\overline{\mathcal{O}}^{(\tau)} = (\overline{\mathbf{o}}^{(\tau)}_1, \overline{\mathbf{o}}^{(\tau)}_2, ..., \overline{\mathbf{o}}^{(\tau)}_K)$, the corresponding likelihood function takes the form:
        \begin{equation}\label{eq: thm_1}
            p(\overline{\mathcal{O}}^{(\tau)}|\boldsymbol{\mu}) = \prod^R_i \prod^K_j \mu_j^{\overline{\mathbf{o}}^{(\tau)}_{ij}} = \prod^K_j \mu_j^{\sum_n\overline{\mathbf{o}}^{(\tau)}_{nj}} = \prod^K_j \mu_j^{m_j}
        \end{equation}
        We see that the likelihood function depends on the $K$ data points only through the $R$ quantities: $m_j = \sum_n\overline{\mathbf{o}}^{(\tau)}_{nj}$. Since the number of observations of $\overline{\mathbf{o}}^{(\tau)}_j$ equals $1$, we achieved sufficient statistics for this distribution. Therefore, $p(\overline{\mathcal{O}}^{(\tau)}) = p(\overline{\mathcal{O}})$ can be proved.
        
        \textit{Meta-path Ratio Preservation in Generation.}
        Since we have proved the meta-path ratio can be preserved during the sampling, the next step is to show that the distribution of generated meta-paths $p(\overline{\mathcal{O}}^{(g)})$ is equal to $p(\overline{\mathcal{O}}^{(\tau)})$. Proving $p(\overline{\mathcal{O}}^{(g)}) = p(\overline{\mathcal{O}}^{(\tau)})$ is equivalent to prove whether $p_{data} = p_g$ in the GAN setting.  As being proved in the works of GANs and their variants \cite{goodfellow2014generative, arjovsky2017wasserstein}, it showed that the objective function of the generator $G$ is equivalent to optimize the distribution distance between $p_{data}$ and $p_g$ if the discriminator $D$ is optimal. Therefore, global optimality of $p_g = p_{data}$ can be achieved if both generator $G$ and discriminator $D$ have enough capability. Therefore, $p(\overline{\mathcal{O}}^{(g)}) = p(\overline{\mathcal{O}}^{(\tau)})$ if both $G$ and $D$ are optimal in our framework. 
        
        \textit{Meta-path Ratio Preservation in Assembling.} Finally, we show that our graph assembling method can also preserve the meta-path ratio from the generated data $\overline{\mathcal{O}}^{(g)}$ such that $p(\overline{\mathcal{O}}^{(g)}) = p(\overline{\mathcal{O}}^{(r)})$. As discussed in Sec. \ref{sec: framework}, the new graph $\hat{\mathcal{G}}$ is directly assembled by meta-paths $(\overline{\mathbf{o}}^{(g)}_1, \overline{\mathbf{o}}^{(g)}_2, ..., \overline{\mathbf{o}}^{(g)}_Q)$ that are sampled $i.i.d$ from $\overline{\mathcal{O}}^{(g)}$ with sampling size $Q$, which is exactly the reverse procedure of Eq. \eqref{eq: thm_1}. 
        
        Therefore, if both generator $G$ and discriminator $D$ are optimal, the multinomial distribution $p(\overline{\mathcal{O}})$ of distinct meta-path patterns can be preserved in all three steps of our generation framework.
    \end{proof}
    
    \section{Experiment}\label{sec: exp}
    In this section, we compare HGEN to the adaption of closest state-of-the-art baselines, demonstrating its effectiveness in generating realistic heterogeneous graphs in diverse settings.
	
	
	\subsection{Data}
	
	\textbf{Synthetic Datasets.} We synthesis random heterogeneous graphs of different sizes through the combination of $N$ overlapping homogeneous graphs, where the overlap is accomplished by node sharing. We generate three random heterogeneous graphs (named as $\text{Syn}_{100}$, $\text{Syn}_{200}$, and $\text{Syn}_{500}$) with node size $100$, $200$, and $500$, respectively. The number of node types in each of the synthetic heterogeneous graph is $3$. 
    
    \textbf{Real-world Datasets.} We also employ three large-scale real-world heterogeneous graph datasets in our experiment.
    \begin{itemize}[leftmargin=*]
        \item \textbf{PubMed}. This dataset consists of four classes of nodes: Gene (G), Disease (D), Chemical (C), and Species (S). We construct a sub-graph that relates to all Chemical nodes labeled in \cite{yang2020heterogeneous}. There are $1,565$ nodes and $13,532$ edges.
        \item \textbf{IMDB}. This movie-related heterogeneous graph is adopted from \cite{wang2019heterogeneous}, which contains three node types: Director (D), Actor (A), Movie (M), and Genre (G). We construct a subgraph that contains all the movies with a score $\ge 7.5$. This graph contains $1,653$ nodes and $4,267$ edges.
        \item \textbf{DBLP}. This heterogeneous graph adopted from \cite{wang2019heterogeneous} contains Paper (P), Author (A), Venue (V), and Term (T) as node types. We sample a subgraph that is related to five computer science venues: \textit{KDD}, \textit{WSDM}, \textit{WWW}, \textit{ICDM}, and \textit{ICML}. There are $1,565$ nodes and $47,885$ edges.
    \end{itemize}
    
    \begin{table*}[t]
\centering
\resizebox{0.93\textwidth}{!}{%
\begin{tabular}{@{}cccccccc|ccc@{}}
\toprule
Graphs                                        & Models          &  LCC & TC & Clustering Coef. & Powerlaw Coef. & Assortativity & Degree Distribution Dist. & EO Rate & Uniqueness \\ \midrule
\multicolumn{1}{c|}{\multirow{7}{*}{Syn-100}} & GraphRNN
&      $78.43\pm 2.23$               &      $16.62\pm 5.42$          &      $0.002\pm0.01$             &        $1.611\pm 0.09$       &     $\mathbf{-0.153\pm 0.07}$          &     2.19e$-2\pm3.21$e$-3$         &       $37.21\%\pm 1.08\%$         &    $33.09\%\pm 7.06\%$         \\ 
\multicolumn{1}{c|}{}                         & NetGAN          &         $80.12 \pm 3.45$            &   $6.79 \pm 1.76$             &         $0.001 \pm 0.00$          &     $1.524 \pm 0.21$          &      $-0.213 \pm 0.09$         &       1.33e$-2\pm6.46$e$-3$       &       $8.74\%\pm 0.82\%$         &    $\mathbf{94.03\%\pm 0.49\%}$         \\
\multicolumn{1}{c|}{}                         & GraphVAE        &    $99.01\pm0.00$                 &       $224.81\pm5.13$         &        $0.70\pm0.04$           &  $4.579\pm0.05$           &       $-0.73\pm0.05$        &      3.71e$-1\pm1.98$e$-2$            &       $11.5\% \pm 1.09\%$         &       $65.54\%\pm 2.98\%$          \\ 
\multicolumn{1}{c|}{}                         & VGAE            &     $48.9\pm4.63$                &       $63.7\pm46.25$         &       $0.184\pm0.06$            &     $\mathbf{1.87\pm0.10}$          &      $0.1\pm0.03$         &           2.23e$-1\pm6.08$e$-2$      &  $\mathbf{3.23\%\pm 0.09\%}$              &    $51.1\%\pm 3.04\%$      \\ 
\multicolumn{1}{c|}{}                         & \textbf{HGEN} &               $\mathbf{81.13\pm2.42}$   &       $\mathbf{53.12\pm3.78}$         &       $\mathbf{0.079\pm0.01}$            &      $1.782\pm0.01$         &      $-0.114\pm0.03$         &        \textbf{8.79e$\mathbf{-3\pm3.12}$e$\mathbf{-3}$}      &    $10.2\%\pm 0.17\%$            &    $92.97\%\pm 0.72\%$         \\ \cmidrule(l){2-10}
\multicolumn{1}{c|}{}                         & \textit{Real}            &          85   &       36         &        0.072           &        1.832       &      -0.169         &        N/A        &       N/A         &     N/A      \\ \midrule \midrule
\multicolumn{1}{c|}{\multirow{7}{*}{Syn-200}} & GraphRNN        &    $132.76\pm 1.08$                 &      $2.54\pm 0.77$          &    $0.001\pm 0.00$               &      $1.603\pm 0.01$         &      $-0.05\pm 0.01$         &    5.15e$-2\pm3.07$e$-3$       &      $25.81\%\pm 2.65\%$          &     $27.72\%\pm 3.07\%$           \\ 
\multicolumn{1}{c|}{}                         & NetGAN          &  $153\pm1.56$                   &      $2.24\pm0.35$          &    $0.001 \pm 0.00$               &      $1.579\pm 0.31$         &      $-0.008\pm 0.001$         &          6.43e$-2\pm4.2$e$-3$       &        $11.32\%\pm 0.77\%$        &   $95.88\%\pm 3.19\%$       \\ 
\multicolumn{1}{c|}{}                         & GraphVAE        &      $195.43\pm1.12$               &     $51.32\pm1.01$           &       $0.002\pm0.001$            &        $5.377\pm0.21$       &      $-0.75\pm 0.05$         &     5.38e$-1\pm1.7$e$-2$      &       $\mathbf{1.78\%\pm 0.41\%}$         &     $64.37\%\pm 2.94\%$           \\  
\multicolumn{1}{c|}{}                         & VGAE            &    $86.2\pm16.93$                 &       $860.4\pm185.9$         &       $0.23\pm0.04$            &     $\mathbf{1.787\pm0.08}$          &       $0.2\pm0.15$        &      8.53e$-2\pm2.14$e$-2$       &     $3.74\%\pm 0.08\%$           &    $59.65\%\pm 1.46\%$          \\ 
\multicolumn{1}{c|}{}                         & \textbf{HGEN} &         $\mathbf{158.5\pm2.64}$            &       $\mathbf{38.5\pm5.26}$          &       $\mathbf{0.043\pm0.01}$            &       $1.732\pm0.02$        &       $\mathbf{-0.065\pm0.04}$        &         \textbf{2.25e$\mathbf{-2\pm5.5}$e$\mathbf{-3}$}     &      $4.22\%\pm 0.67\%$          &     $\mathbf{96.31\%\pm 5.11\%}$       \\ \cmidrule(l){2-10} 
\multicolumn{1}{c|}{}                         & \textit{Real}            &        180             &       28         &         0.037          &        1.809       &       -0.089        &     N/A        &        N/A        &   N/A           \\ \midrule \midrule
\multicolumn{1}{c|}{\multirow{6}{*}{Syn-500}} & GraphRNN        &     $311.59\pm 2.14$                &       $11.53\pm 5.57$         &      $0.004\pm 0.001$             &     $1.862\pm 0.01$          &      $1.862\pm0.002$         &       4.05e$-2\pm1.1$e$-3$         &    $21.87\%\pm 0.86\%$            &    $29.54\%\pm 4.32\%$       \\  
\multicolumn{1}{c|}{}                         & NetGAN          &         $305.81\pm14.28$            &      $\mathbf{3\pm1.21}$          &      $\mathbf{0.001\pm0.001}$             &      $1.812\pm0.07$         &      $0.03\pm0.12$         &       4.83e$-2\pm7.4$e$-4$       &     $6.72\% \pm 0.13\%$           &     $93.98\%\pm 0.21\%$        \\ 
\multicolumn{1}{c|}{}                         & VGAE            &     $97.0\pm29.24$                &        $4346.2\pm453.62$        &     $0.193\pm0.02$              &       $1.77\pm0.06$        &      $-0.022\pm0.09$         &       2.22e$-1\pm2.4$e$-2$       &    $5.46\%\pm 1.12\%$            &     $63.65\%\pm3.1\%$        \\ 
\multicolumn{1}{c|}{}                         & \textbf{HGEN} &         $\mathbf{347.88\pm7.63}$            &      $74.88\pm4.78$          &      $0.031\pm0.01$             &      $\mathbf{1.865\pm0.02}$         &      $\mathbf{-0.097\pm0.01}$         &       \textbf{2.81e$\mathbf{-2\pm3.4}$e$\mathbf{-3}$}       &      $\mathbf{1.49\%\pm 0.11\%}$          &    $\mathbf{95.89\%\pm 1.18\%}$         \\ \cmidrule(l){2-10} 
\multicolumn{1}{c|}{}                         & \textit{Real}            &         417            &       8         &         6.5e$-3$          &       1.978        &      -0.12         &      N/A       &       N/A         &    N/A          \\ \midrule \midrule
\multicolumn{1}{c|}{\multirow{6}{*}{PubMed}} & GraphRNN        &           $1563.23\pm 32.46$          &       $1549.79\pm 33.62$         &          $0.01\pm 0.007$         &         $\mathbf{1.753\pm 0.04}$      &      $-0.03\pm 0.01$         &     1.61e$-1\pm3.71$e$-2$      &        $13.41\%\pm 1.24\%$        &      $54.62\%\pm 4.32\%$          \\ 
\multicolumn{1}{c|}{}                         & NetGAN          &            $793.2\pm 41.5$         &      $18.3\pm0.9$          &   $0.001\pm0.00$                &     $1.47\pm 0.11$          & $-0.12\pm0.02$              &      6.69e$-2\pm1.5$e$-3$          &        $4.32\% \pm 0.54\%$        &     $78.03\% \pm 0.19\%$      \\
\multicolumn{1}{c|}{}                         & VGAE            &    $347.9\pm7.03$                 &       $70,982.2\pm 4,086.53$         &       $0.234\pm0.01$            &       $2.48\pm0.01$        &     $-0.466\pm0.01$          &     1.38e$-1\pm4.8$e$-3$       &    $\mathbf{\approx 0\%}$            &    $22.87\% \pm 1.68\%$           \\ 
\multicolumn{1}{c|}{}                         & \textbf{HGEN} &            $\mathbf{825.6\pm 22.1}$         &      $\mathbf{1569.3\pm31.3}$          &   $\mathbf{0.034\pm0.003}$                &     $1.634\pm 0.07$          & $\mathbf{-0.143\pm0.08}$              &      \textbf{3.92e$\mathbf{-2\pm7.5}$e$\mathbf{-4}$}      &       $0.07\%\pm 0.01\%$         &     $\mathbf{93.91\% \pm 0.12\%}$          \\ \cmidrule(l){2-10} 
\multicolumn{1}{c|}{}                         & \textit{Real}            &            $948$         &        $2,114$        &        $0.068$           &         $1.75$      &       $-0.208$        &    N/A      &       N/A         &     N/A            \\ \midrule \midrule
\multicolumn{1}{c|}{\multirow{6}{*}{IMDB}} & GraphRNN        &              $1425.47\pm 121.5$       &      $142.13\pm5.87$          &     $0.179\pm0.02$              &      $2.97\pm0.05$         &       $0.05\pm0.04$        &      1.98e$-1\pm2.61$e$-3$         &     $9.87\%\pm 0.51\%$           &       $21.52\% \pm 3.31\%$      \\ 
\multicolumn{1}{c|}{}                         & NetGAN          &             $932.5\pm8.49$        &        $\mathbf{0.0\pm0.0}$        &   $\mathbf{0.0\pm0.0}$               &      $2.08\pm0.01$         &          $\mathbf{-0.25\pm0.07}$     &   1.36e$-1\pm1.89$e$-3$        &   $7.62\%\pm 0.07\%$             &     $82.69\% \pm 1.27\%$          \\ 
\multicolumn{1}{c|}{}                         & VGAE            &    $635.2\pm4.16$                 &     $7,752.4\pm281.32$           &       $0.141\pm0.01$            &      $2.02\pm0.02$         &     $-0.49\pm0.15$          &   1.9e$-1\pm2.33$e$-3$        &      $\mathbf{\approx 0\%}$          &      $42.71\% \pm 1.47\%$          \\ 
\multicolumn{1}{c|}{}                         & \textbf{HGEN} &             $\mathbf{945.2\pm11.54}$        &        $26.0\pm3.28$        &    3.56e$-3\pm3.42$e$-4$               &      $\mathbf{2.16\pm0.01}$         &          $-0.19\pm0.04$     &   \textbf{4.36e$\mathbf{-2\pm4.25}$e$\mathbf{-4}$}        &    $2.69\%\pm 0.04\%$            &       $\mathbf{88.71\% \pm 0.39\%}$        \\      \cmidrule(l){2-10}   \multicolumn{1}{c|}{}           & Real            &           $1,074$          &      1          &        4.43e$-4$           &     $2.51$          &    -0.235           &    N/A         &       N/A         &      N/A        \\ \midrule \midrule
\multicolumn{1}{c|}{\multirow{5}{*}{DBLP}}  & NetGAN          &               $10,353\pm72.71$      &     $\mathbf{0.0\pm0.0}$           &      $\mathbf{0.0\pm0.0}$             &         $3.308\pm0.41$      &          $-0.059\pm0.03$     &      5.03e$-1\pm2.1$e$-2$       &         $5.48\%\pm 0.32\%$       &    $\mathbf{72.51\% \pm 0.32\%}$          \\ 
\multicolumn{1}{c|}{}                         & VGAE            &   $3,771\pm 236.29$                  &        $1214.69\pm 452.61$        &        $0.271\pm0.06$           &      $1.579\pm 0.07$         &   $-0.44\pm0.11$            &      8.71e$-2\pm1.77$e$-3$       &     $\mathbf{\approx 0\%}$           &      $17.26\% \pm 0.41\%$        \\  
\multicolumn{1}{c|}{}                         & \textbf{HGEN} &               $\mathbf{5,163\pm21.41}$      &     $1068\pm12.83$           &      $0.018\pm 0.001$             &         $\mathbf{1.793\pm0.21}$      &          $\mathbf{-0.157\pm0.03}$     &      \textbf{5.82e$\mathbf{-3\pm1.67}$e$\mathbf{-4}$}      &     $1.55\% \pm 0.09\%$           &     $66.59\% \pm 0.17\%$          \\ \cmidrule(l){2-10}
\multicolumn{1}{c|}{}                         & \textit{Real}            &   $5,513$                  &    $0.0$            &          $0.0$         &  $1.855$             &      $-0.201$         &    N/A       &     N/A           &    N/A            \\ \midrule \midrule
\end{tabular}
}
\caption{Performance evaluation over compared baselines. The \textit{Real} rows include the values of real graphs, while the rest are the evaluation results of different algorithms. The best performance (the closest to real value) achieved under each metric for a particular dataset is highlighted in bold font. Note that we do not include GraphVAE in datasets with ($\ge 500$) nodes and GraphRNN in datasets with ($\ge 10,000$) nodes because the programs return errors.}
\vspace{-8mm}
\label{tab: statistical_graph_metrics}
\end{table*}

\subsection{Experiment Setting}
	In our experiment, we focus on meta-paths with length $1$, $2$, and $3$ as they are the most common ones in heterogeneous graphs \cite{sun2011pathsim}. We sample $10$ graphs from each of the trained models and report their average results and standard deviation in Table \ref{tab: statistical_graph_metrics}. We randomly select $60\%$ of the edges for training, and the remaining graph is used for testing.
	
	\textbf{Baselines.} Since no baseline is available for the novel task of heterogeneous graph generation, we carefully adapt four state-of-the-art graph generation methods: NetGAN \cite{bojchevski2018netgan}, GraphVAE \cite{simonovsky2018graphvae}, VGAE \cite{kipf2016variational}, and GraphRNN \cite{you2018graphrnn}. We utilize node type information as node features of the input graph in GraphVAE and VGAE. In addition, we modify NetGAN and GraphRNN to make them available to generate node types. 
	
	\textbf{Evaluation Metrics.} The evaluation of heterogeneous graph generation can be divided into three aspects. 1) \textit{Graph Statistical Properties}: we focus on six typical statistics as widely used in \cite{bojchevski2018netgan, yang2019conditional, goyal2020graphgen} for measuring the structural similarity, including LCC (the size of the largest connected component), TC (Triangle count), Clustering Coef. (clustering coefficient); Powerlaw Coef. (power-law distribution of the node degree distribution), Assortativity, and Degree Distribution Dist. (Node degree distribution Maximum Mean Discrepancy distance).  2) \textit{Graph Novelty and Uniqueness.}  Ideally, we would want the generated graphs to be diverse and similar, but not identical. To quantify this aspect, we check the uniqueness between the generated graphs by calculating their edit distances. Additionally, we calculate the EO Rate (edge overlapping rate) between the generated graphs and the testing graphs for measuring the novelty of the generated graphs. 3) \textit{Meta-path Ratio Properties}: We measure the preservation of meta-path distribution in two metrics. Firstly, we measure the meta-path length ratio preservation. Secondly, under different meta-path lengths, we also measure the distribution of the frequent meta-path patterns.
    
    \begin{figure*}[!t]
		\subfloat[Meta-path Length Ratio]{\label{fig: path_ratio_syn_500}
			\includegraphics[width=0.24\textwidth]{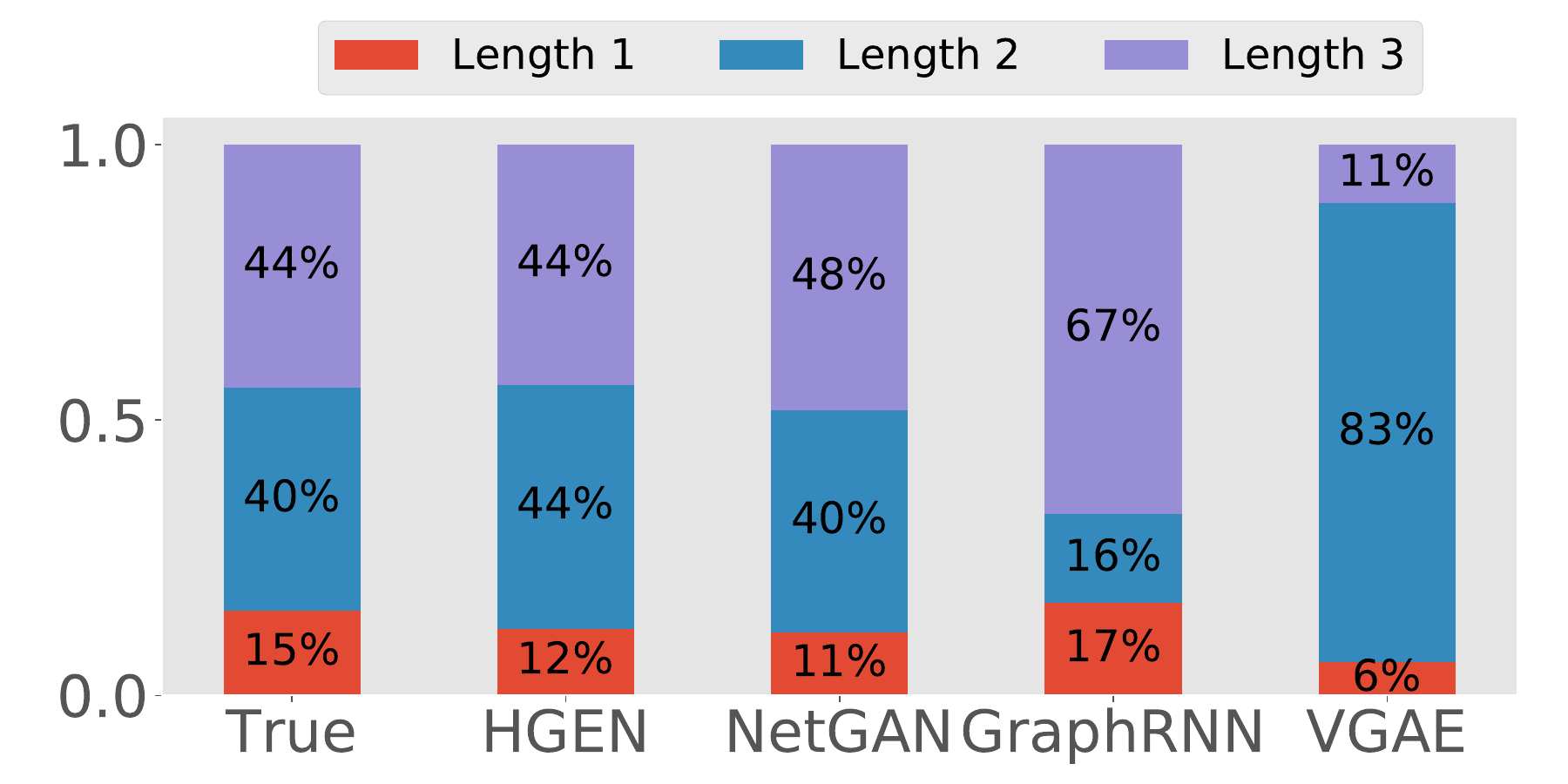}}
		\subfloat[Meta-path Patterns - Length 1]{\label{fig: frequent_path_ratio_syn_500_len_2}
			\includegraphics[width=0.24\textwidth]{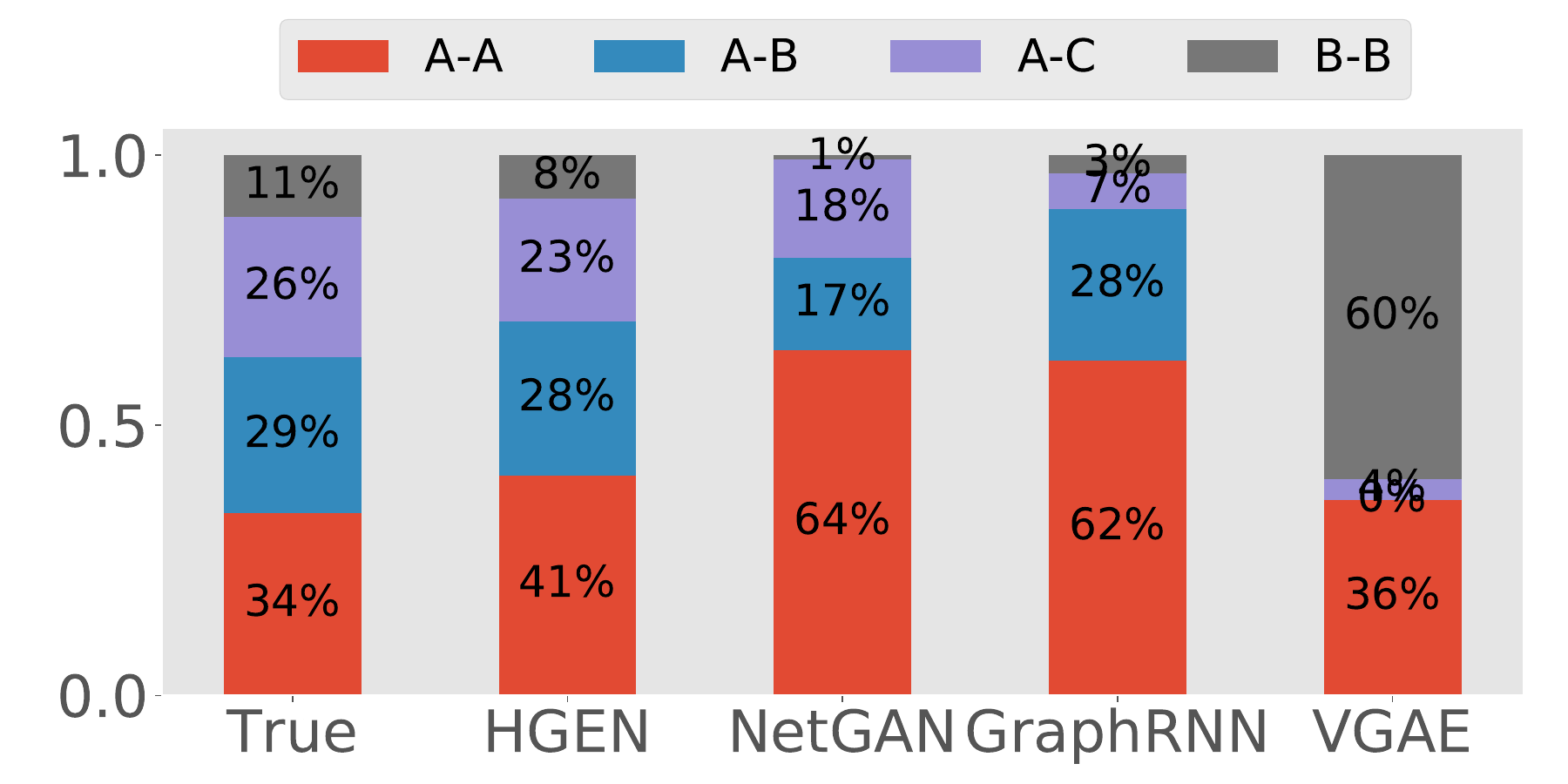}}
		\subfloat[Meta-path Patterns - Length 2]{\label{fig: frequent_path_ratio_syn_500_len_3}
			\includegraphics[width=0.24\textwidth]{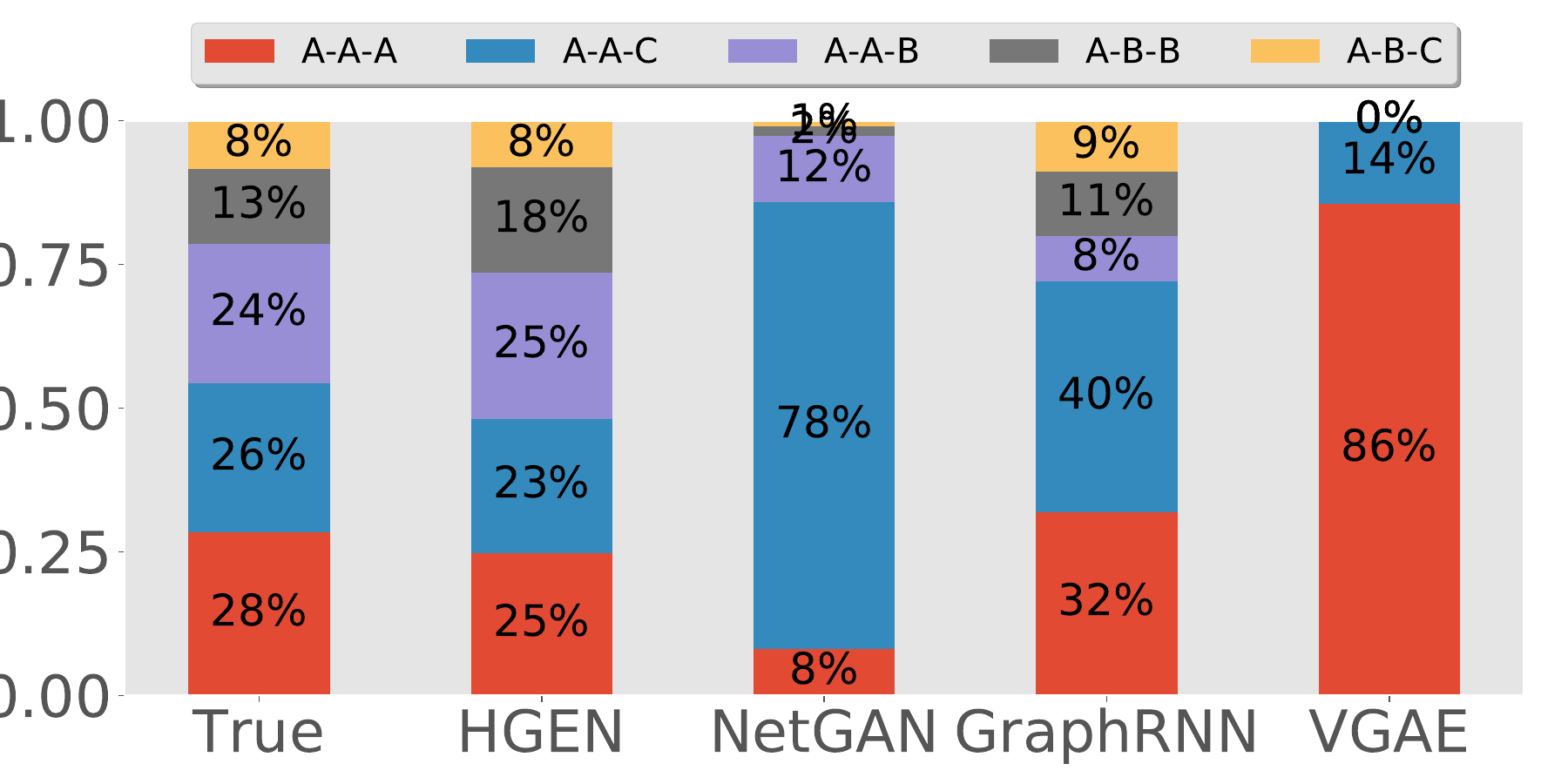}}
		\subfloat[Meta-path Patterns - Length 3]{\label{fig: frequent_path_ratio_syn_500_len_4}
			\includegraphics[width=0.24\textwidth]{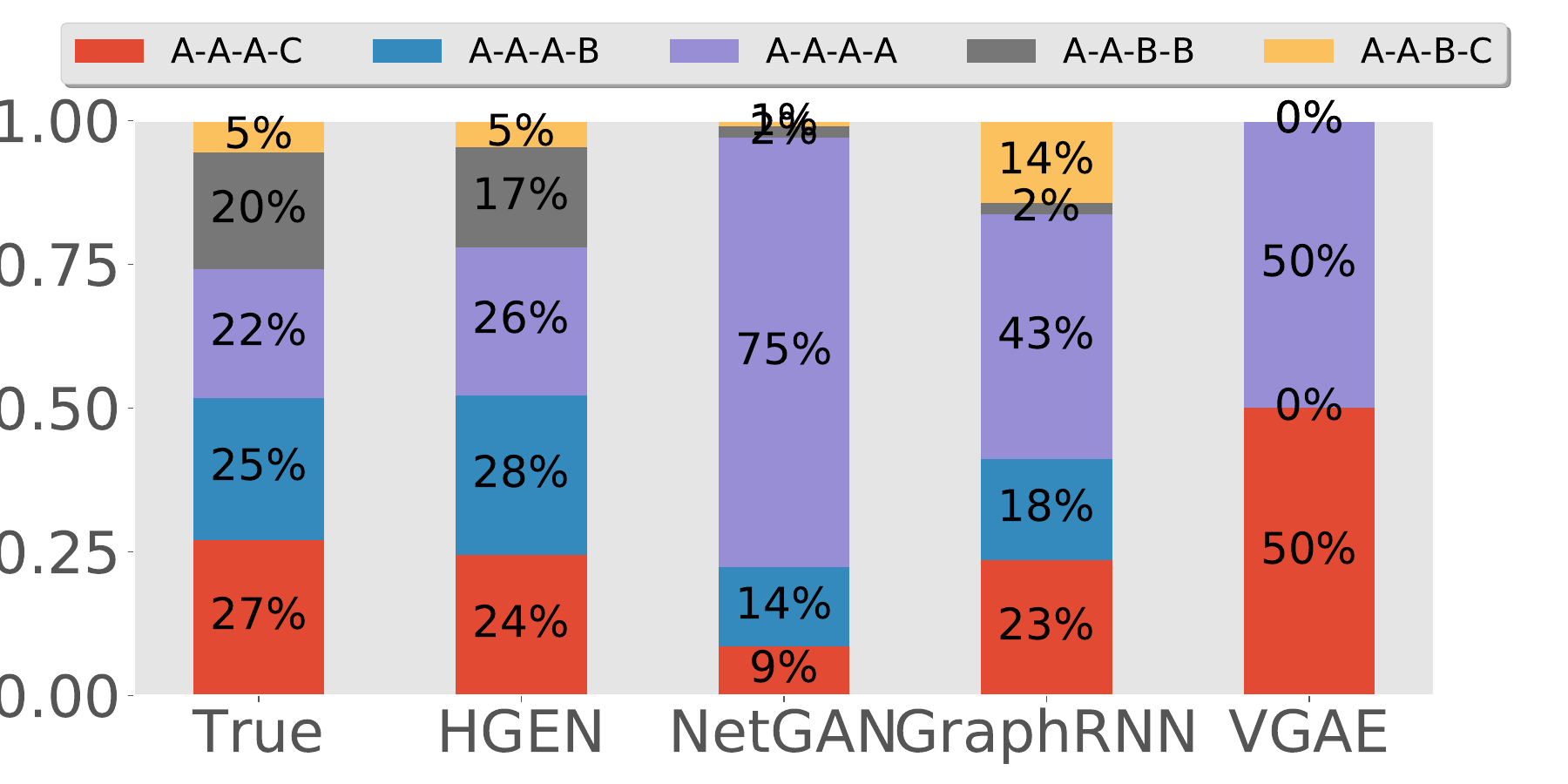}}
		\vspace{-3mm}
		\subfloat[Meta-path Length Ratio]{\label{fig: path_ratio_pubmed}
			\includegraphics[width=0.24\textwidth]{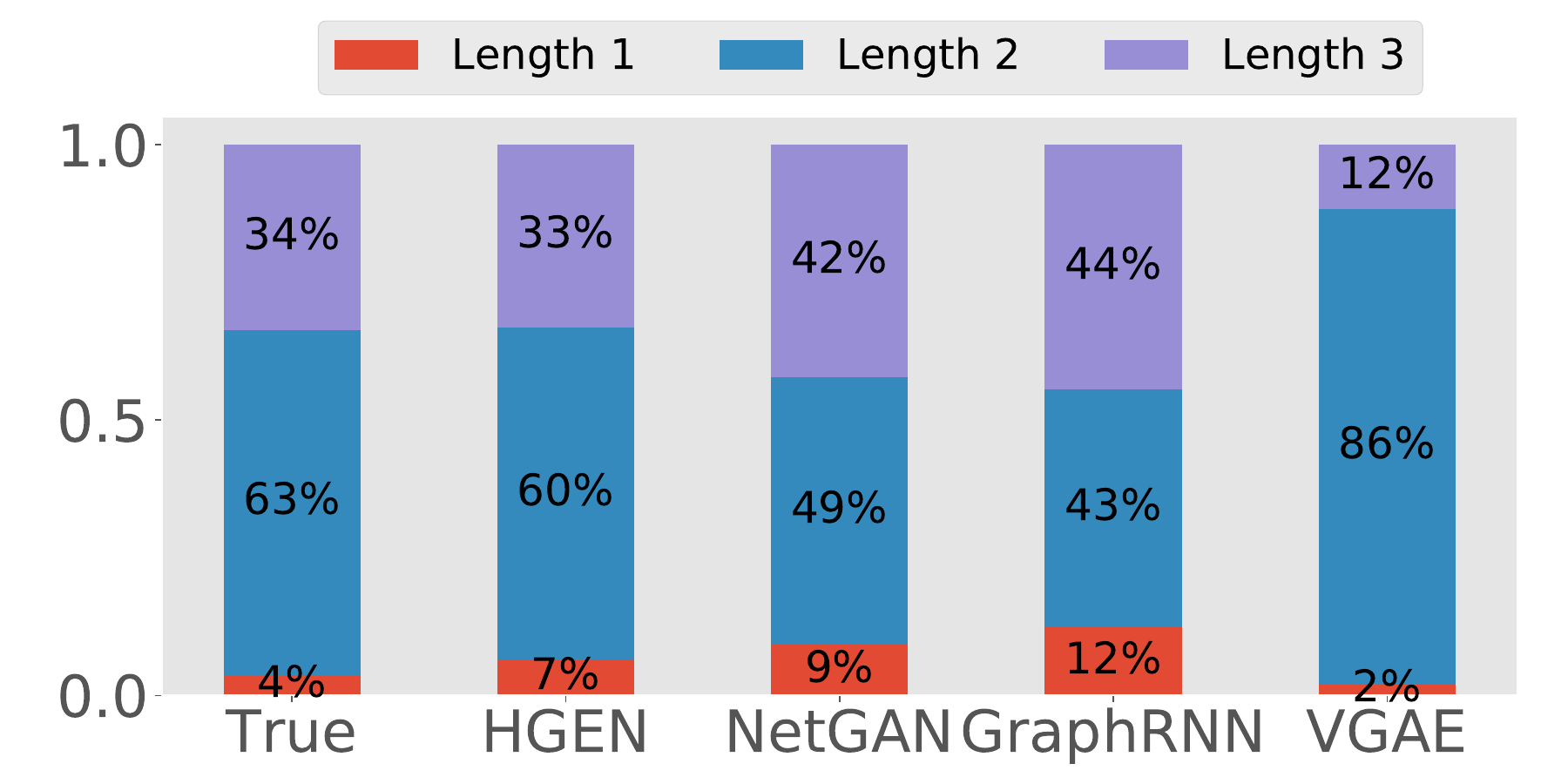}}
		\subfloat[Meta-path Patterns - Length 1]{\label{fig: frequent_path_ratio_pubmed_len_2}
			\includegraphics[width=0.24\textwidth]{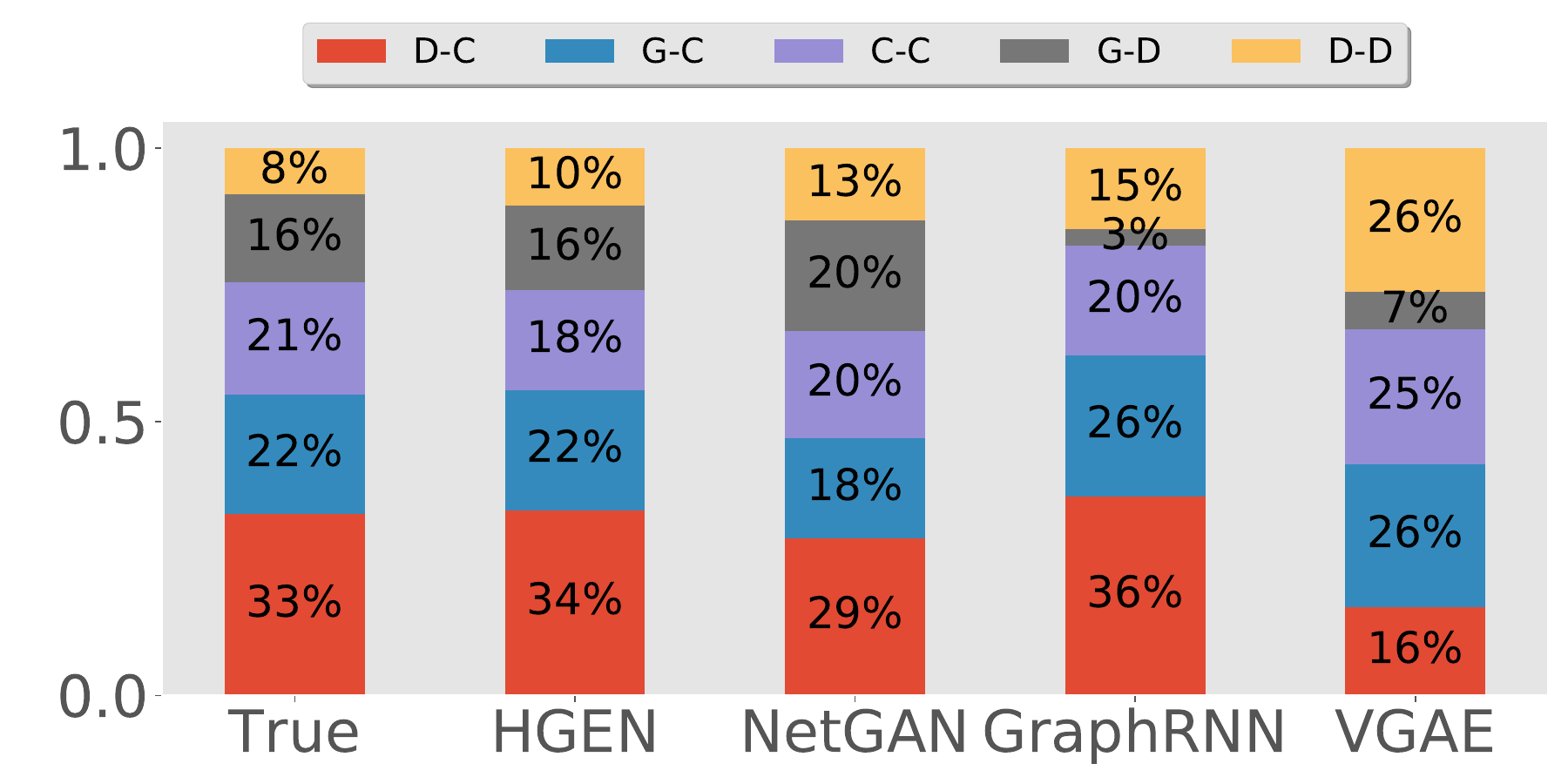}}
		\subfloat[Meta-path Patterns - Length 2]{\label{fig: frequent_path_ratio_pubmed_len_3}
			\includegraphics[width=0.24\textwidth]{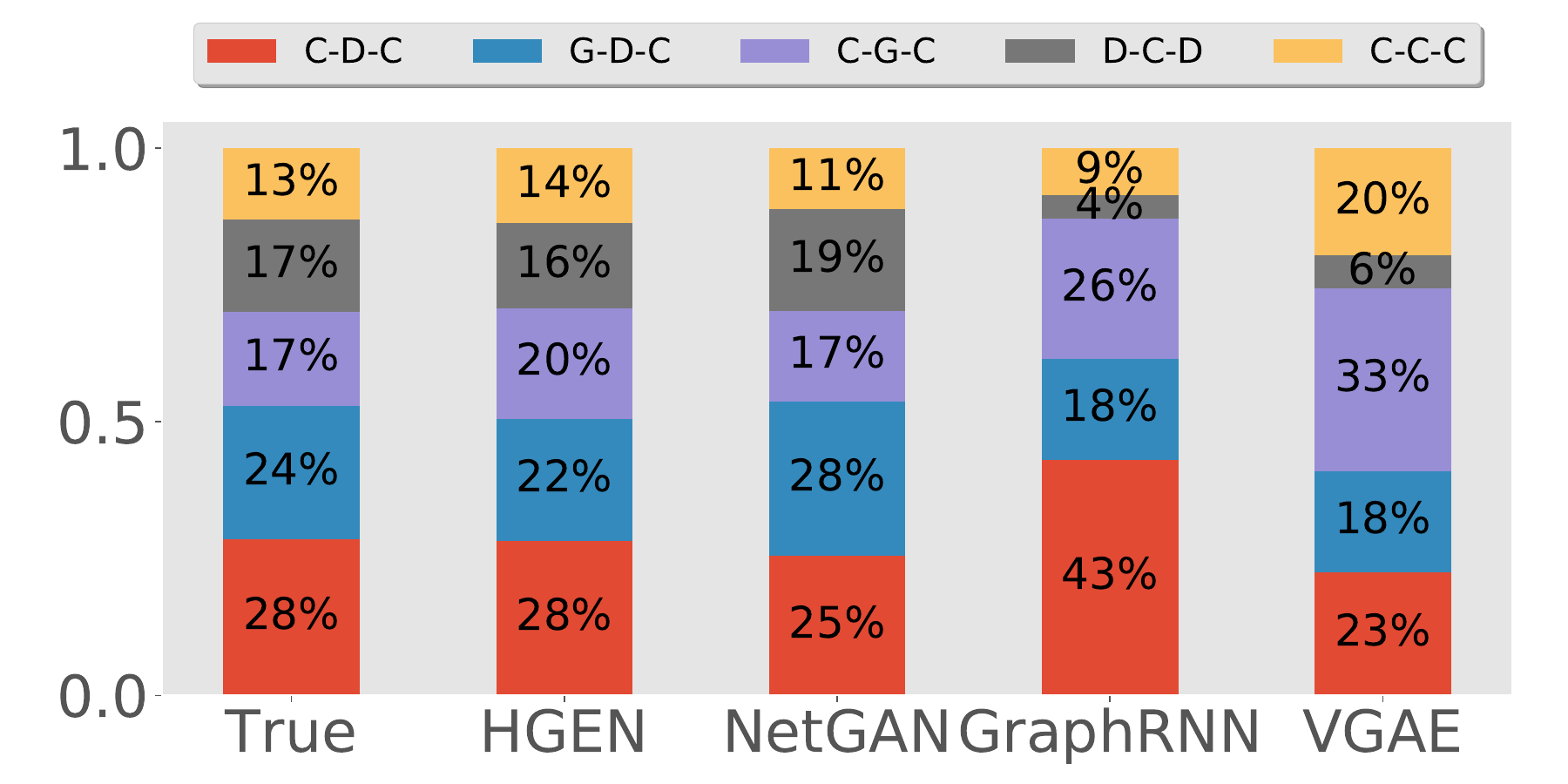}}
		\subfloat[Meta-path Patterns - Length 3]{\label{fig: frequent_path_ratio_pubmed_len_4}
			\includegraphics[width=0.24\textwidth]{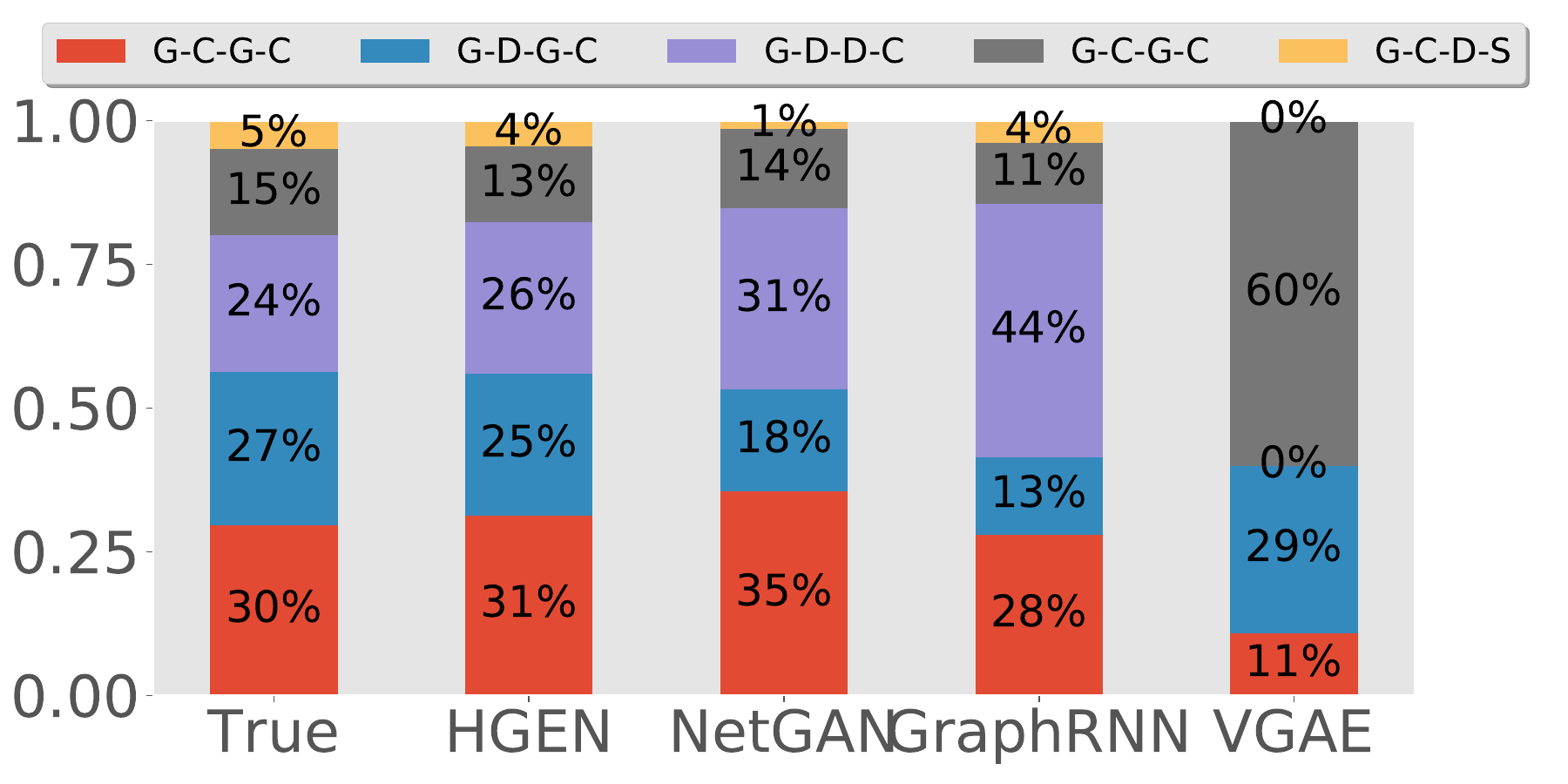}}
		\caption{The meta-path distribution comparison. \ref{fig: path_ratio_syn_500} - \ref{fig: frequent_path_ratio_syn_500_len_4} and \ref{fig: path_ratio_pubmed} - \ref{fig: frequent_path_ratio_pubmed_len_4} are the generated meta-path length distribution along with frequent meta-path patterns distribution with length 1, 2, 3 for Syn\_500 dataset and PubMed dataset, respectively.}
		\label{fig: meta-path-preservation}
		\vspace{-5mm}
	\end{figure*}
	
	\subsection{Quantitative Analysis}
    \textbf{Preservation of Graph Statistical Properties.}
     We evaluate the performance of HGEN against all the baselines on the standard graph statistics, and the results are shown in Table \ref{tab: statistical_graph_metrics}. Overall, HGEN achieves competitive performance constantly with very few exceptions on all metrics over both synthetic and real-world datasets. We report several observations from the table: 1) \textit{Node-level similarity}: HGEN is the dominant performer in most node-level metrics. Although there are no significant differences in both Assortativity and Power-law Coef. among all the algorithms, HGEN rank top with very few exceptions in the node degree distribution distance with at least $40\%$ improvement, which indicates that HGEN can effectively capture the degree distribution of all types of nodes through jointly learning both meta-path and random walk distribution. 2) \textit{Graph level similarity}: HGEN still exceeds other baselines by effectively preserving the community distribution. Specifically, for all the datasets with rich local community information (e.g., PubMed and synthetic datasets), HGEN can utilize the heterogeneous node embedding for preserving the higher-order structural information in the generated heterogeneous walks, which leads to better performance in metrics like LCC, TC, and Clustering Coef.. However, in heterogeneous graphs with rare high-order structures, the performance of HGEN is comparatively less impressive. 3) As shown in Table \ref{tab: statistical_graph_metrics}, the random-walk based method HGEN and NetGAN can generally achieve stable performance than one-shot based (e.g., VGAE and GraphVAE) and sequential-based (GraphRNN) generative models across all datasets. The reason is that random walk based methods learn the overall graph distribution by learning the distribution of its discrete random walks, which is not sensitive to various graph characteristics.  4) Table \ref{tab: statistical_graph_metrics} also shows that VGAE cannot produce realistic graphs even though it achieves the best performance in some metrics, which is expected since the primary purpose of VGAE is learning node embeddings but not generating entire graphs. In addition, as the size of the graph increases, GraphRNN also fails to generate realistic graphs because of the weak scalability of auto-regressive models. 
    
    \begin{table}[t]
\centering
\resizebox{0.47\textwidth}{!}{%
\begin{tabular}{@{}c|cccc|c@{}}
\toprule
                  & HGEN-S  & HGEN-E & HGEN-A & HGEN   & Real   \\ \midrule
LCC               & 1563.76 & 824.14 & 819.32 & $\mathbf{825.6}$  & 948    \\
TC                & 1453.23 & 784.34 & 863.53 & $\mathbf{1569.3}$ & 2114   \\
Clustering Coef.  & 0.026   & 0.015  & 0.016  & $\mathbf{0.034}$  & 0.068  \\
Power Law Coef.   & 1.649    & $\mathbf{1.652}$   & 1.621   & 1.634  & 1.75   \\
Assortativity     & -0.09   & -0.132 & -0.131 & $\mathbf{-0.143}$ & -0.208 \\
Node Degree Dist. & $\mathbf{0.0354}$   & 0.0388   & 0.0515  & 0.0392 & N/A    \\ \bottomrule
\end{tabular}%
}
\caption{Ablation Study in PubMed Dataset}
\label{tbb: ablation}
\vspace{-8mm}
\end{table}

\textbf{Graph Novelty and Uniqueness.} The results of graph novelty and uniqueness are reported at the right two columns in Table \ref{tab: statistical_graph_metrics}. Specifically, HGEN achieves a generally lower EO rate across all datasets, indicating that HGEN does not purely memorize the seen heterogeneous walks in the training data. In contrast, GraphRNN has a higher EO rate, indicating GraphRNN regenerates graphs it saw during training. In addition, VGAE achieves the lowest EO rate since it fails to generate realistic heterogeneous graphs. For Uniqueness, HGEN also exceeds other one-shot and sequential based algorithms by an evident margin, which demonstrates the diversity of the generated graphs. 

    \begin{figure*}[!t]
		\subfloat[Real]{\label{fig: syn_200_real}
			\includegraphics[width=0.16\textwidth]{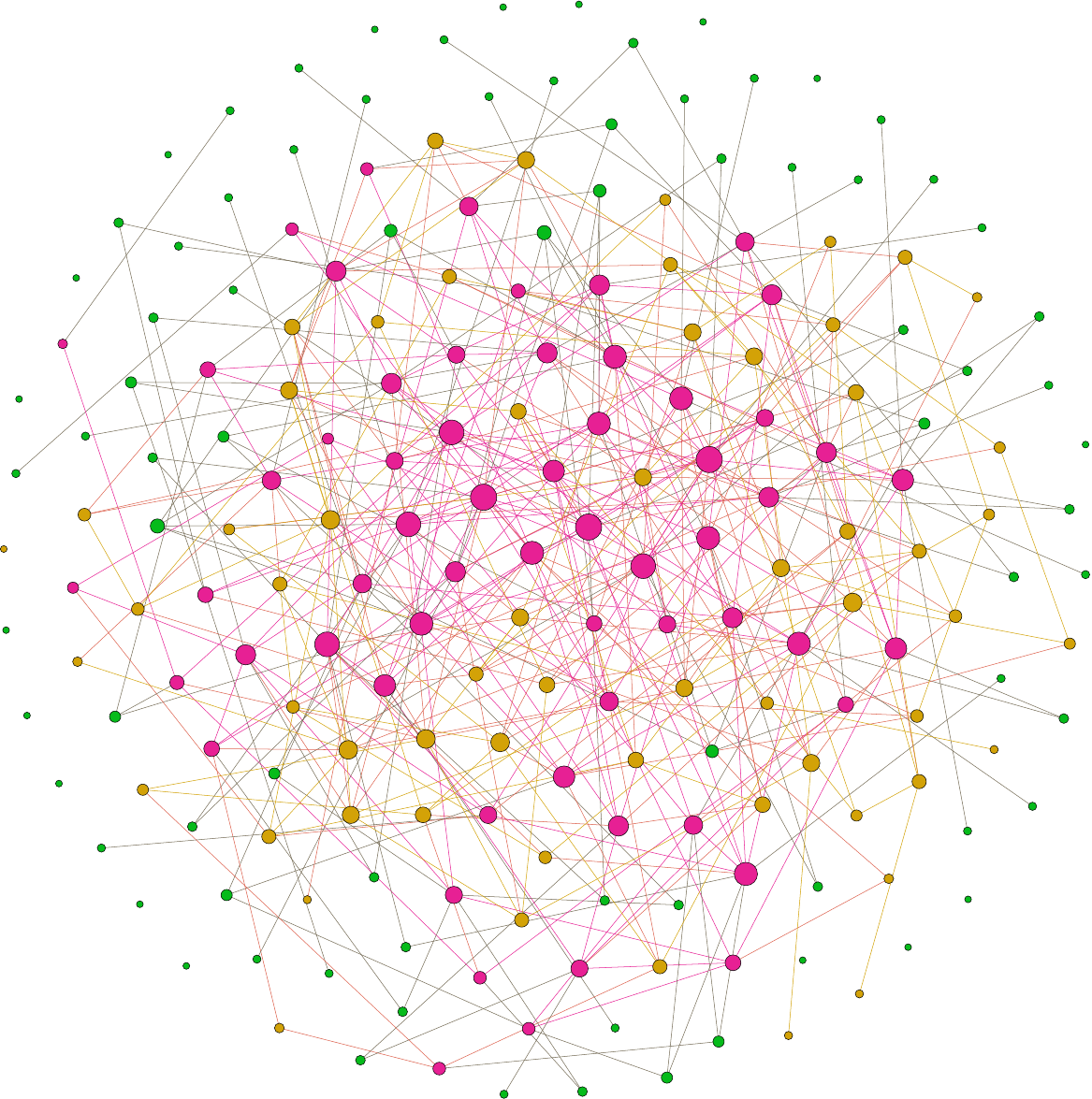}}
		\subfloat[NetGAN]{\label{fig: syn_200_netgan}
			\includegraphics[width=0.16\textwidth]{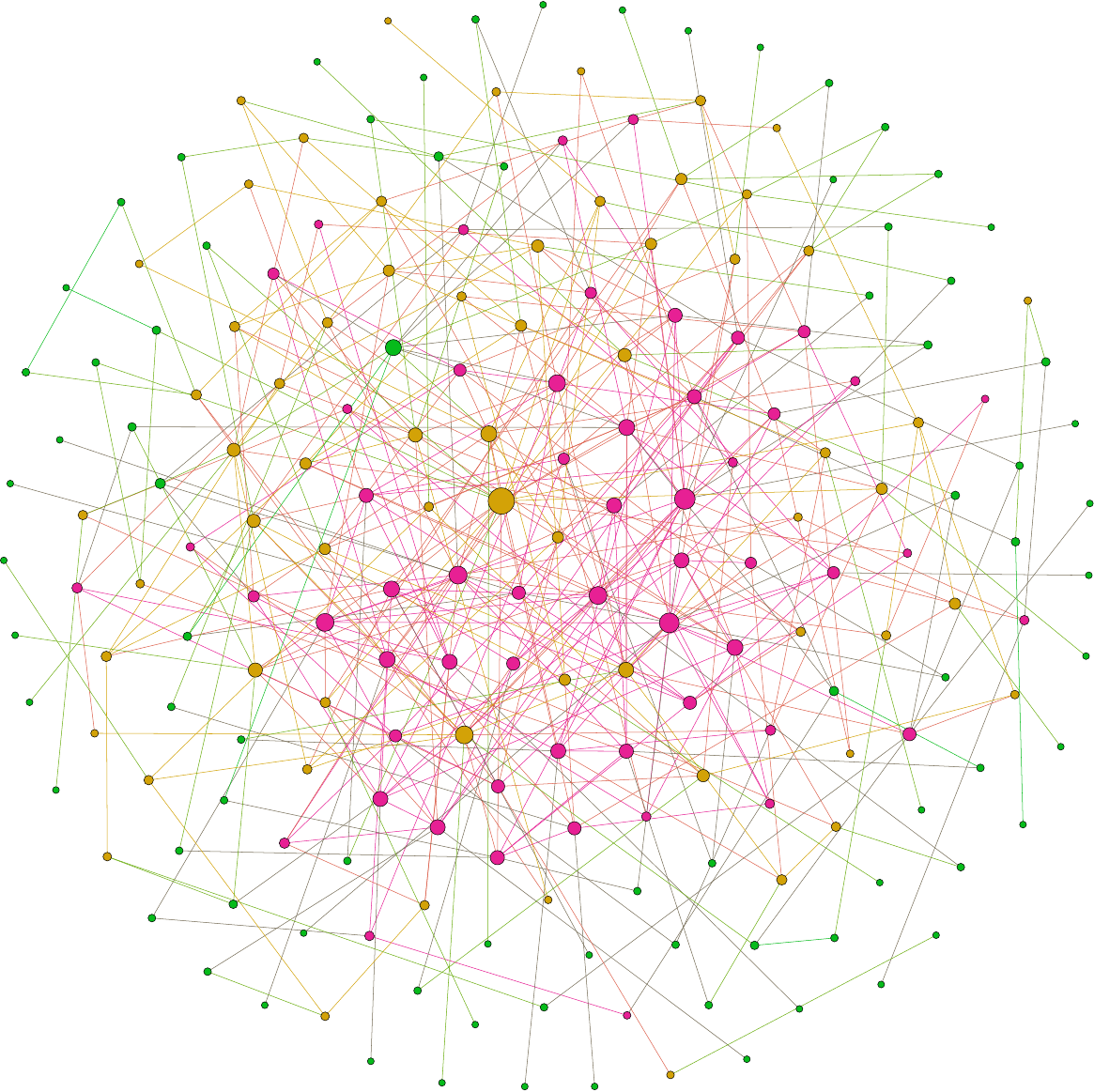}}
		\subfloat[GraphRNN]{\label{fig: syn_200_graphrnn}
			\includegraphics[width=0.16\textwidth]{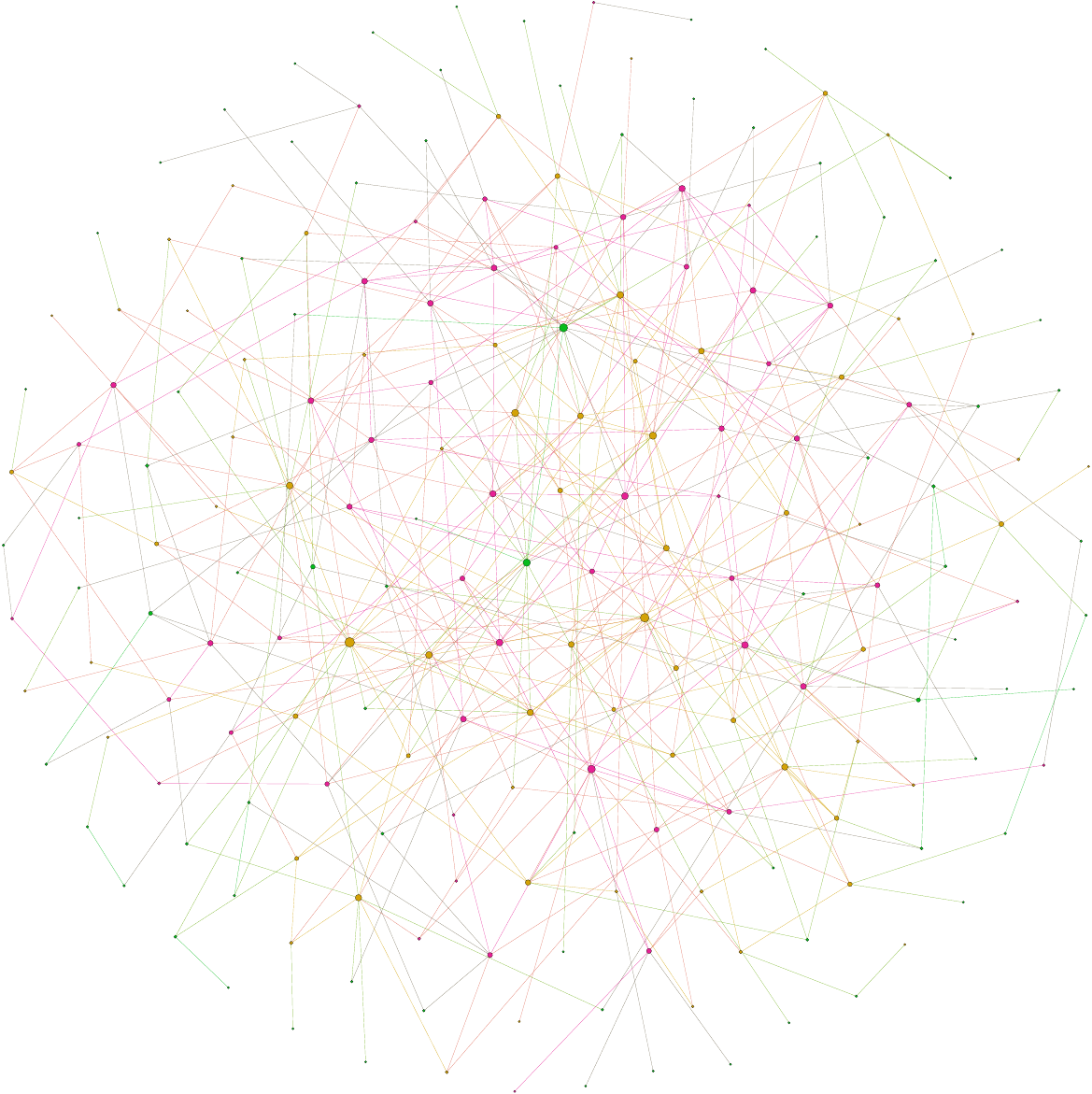}}
		\subfloat[GraphVAE]{\label{fig: syn_200_graphvae}
			\includegraphics[width=0.16\textwidth]{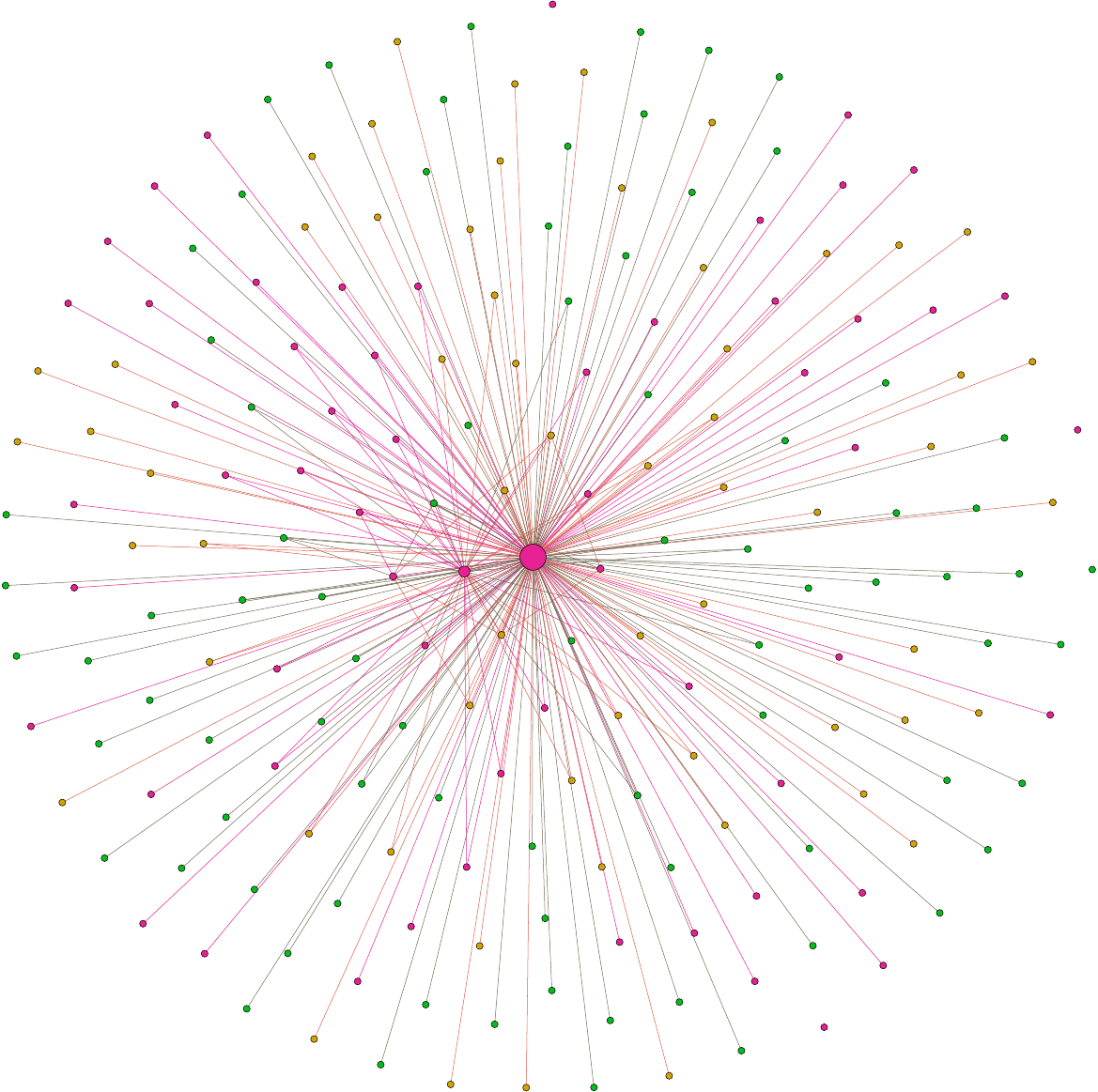}}
		\subfloat[VGAE]{\label{fig: syn_200_vgae}
			\includegraphics[width=0.16\textwidth]{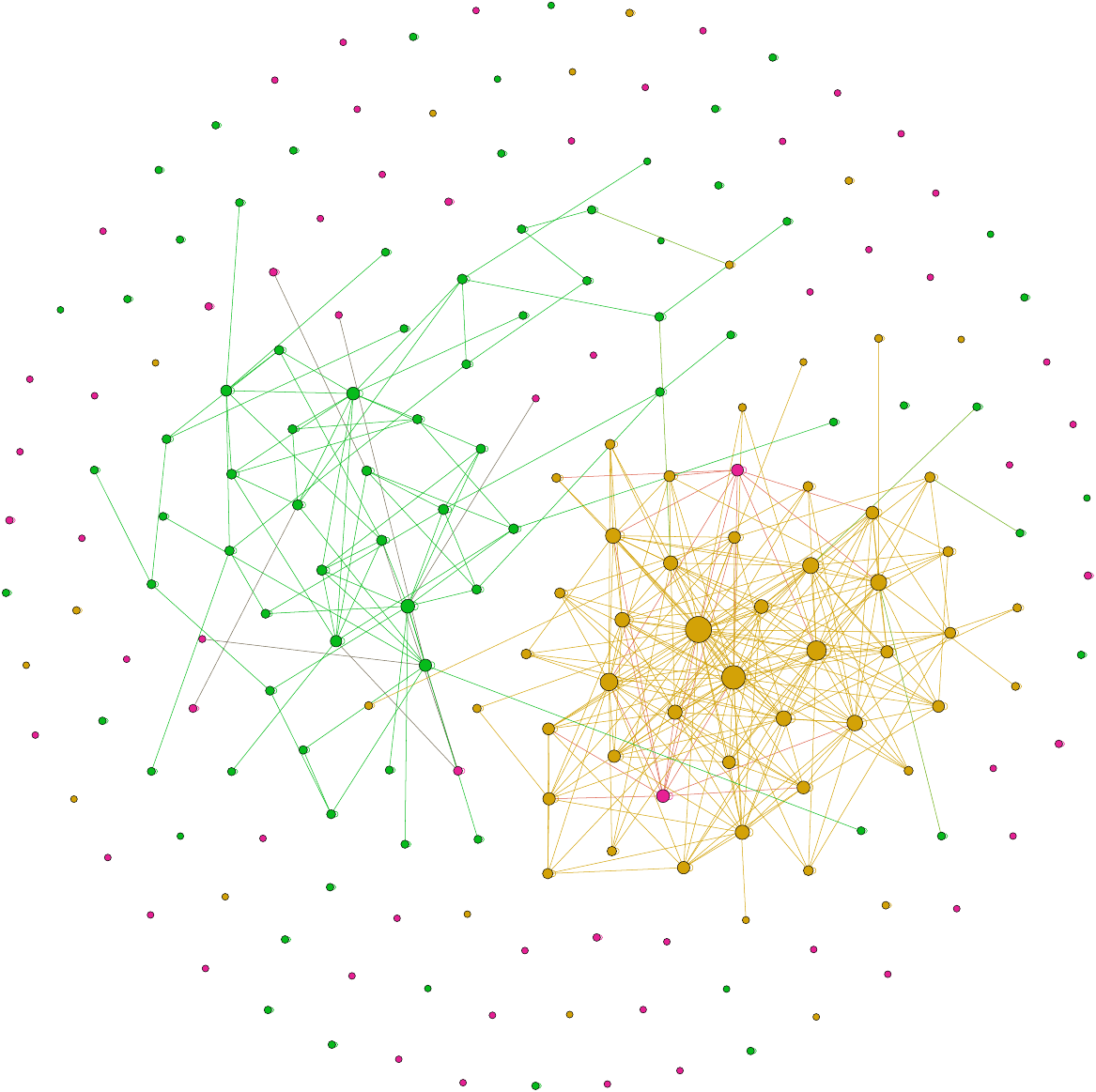}}
		\subfloat[HGEN]{\label{fig: syn_200_hgen}
			\includegraphics[width=0.16\textwidth]{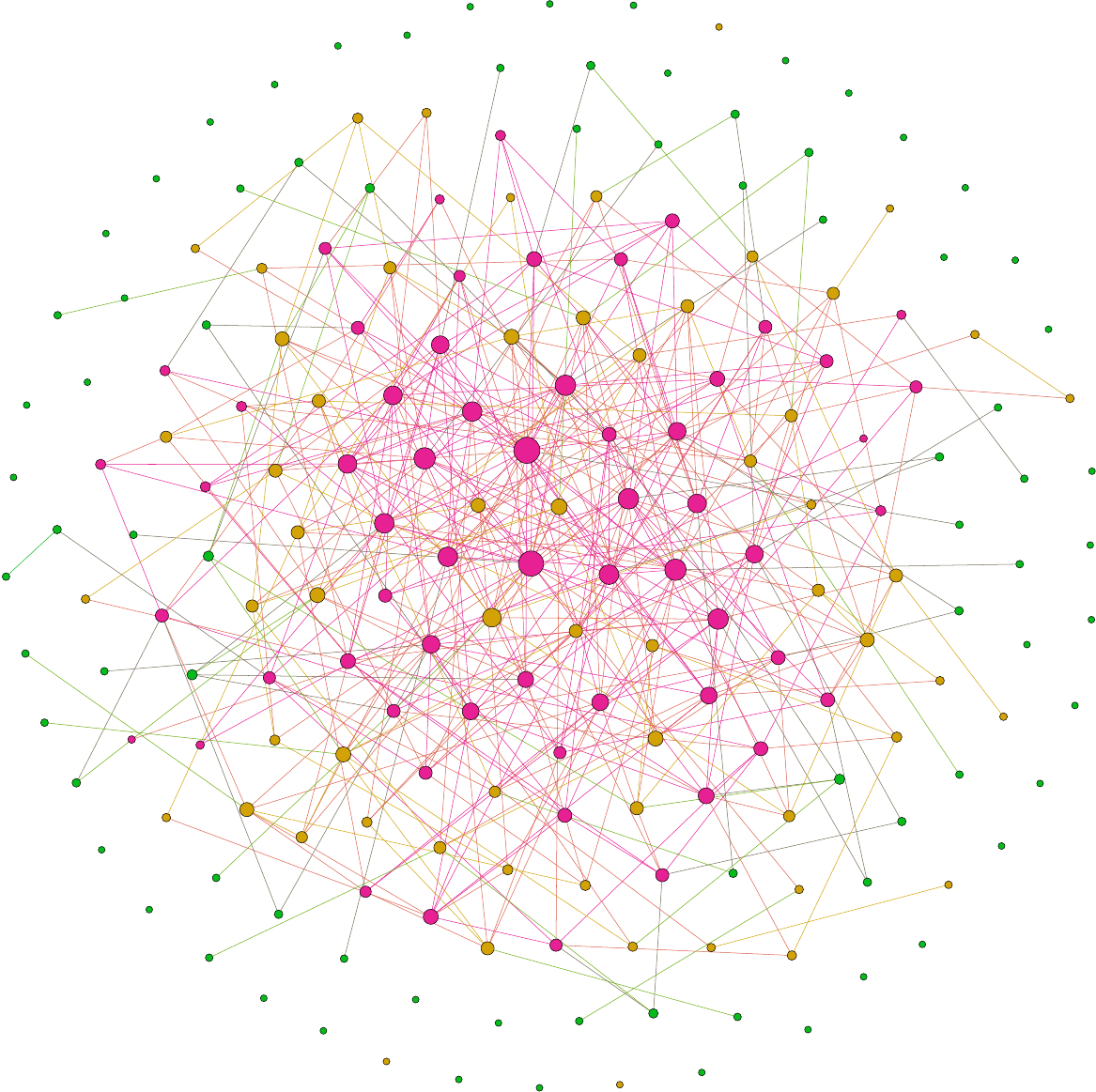}}
    \vspace{-3mm}
		\subfloat[Real]{\label{fig: pubmed_real}
			\includegraphics[width=0.19\textwidth]{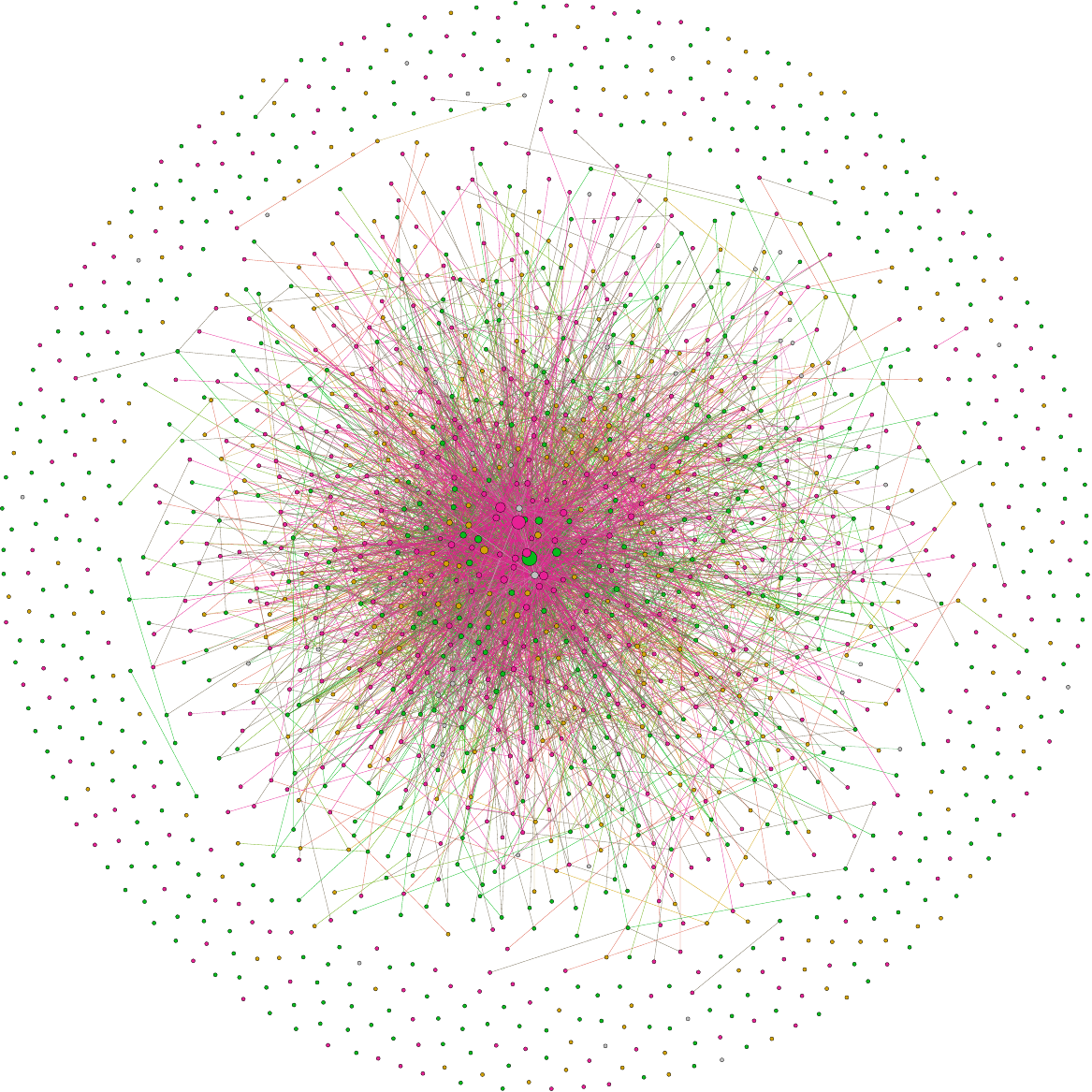}}
		\subfloat[NetGAN]{\label{fig: pubmed_netgan}
			\includegraphics[width=0.19\textwidth]{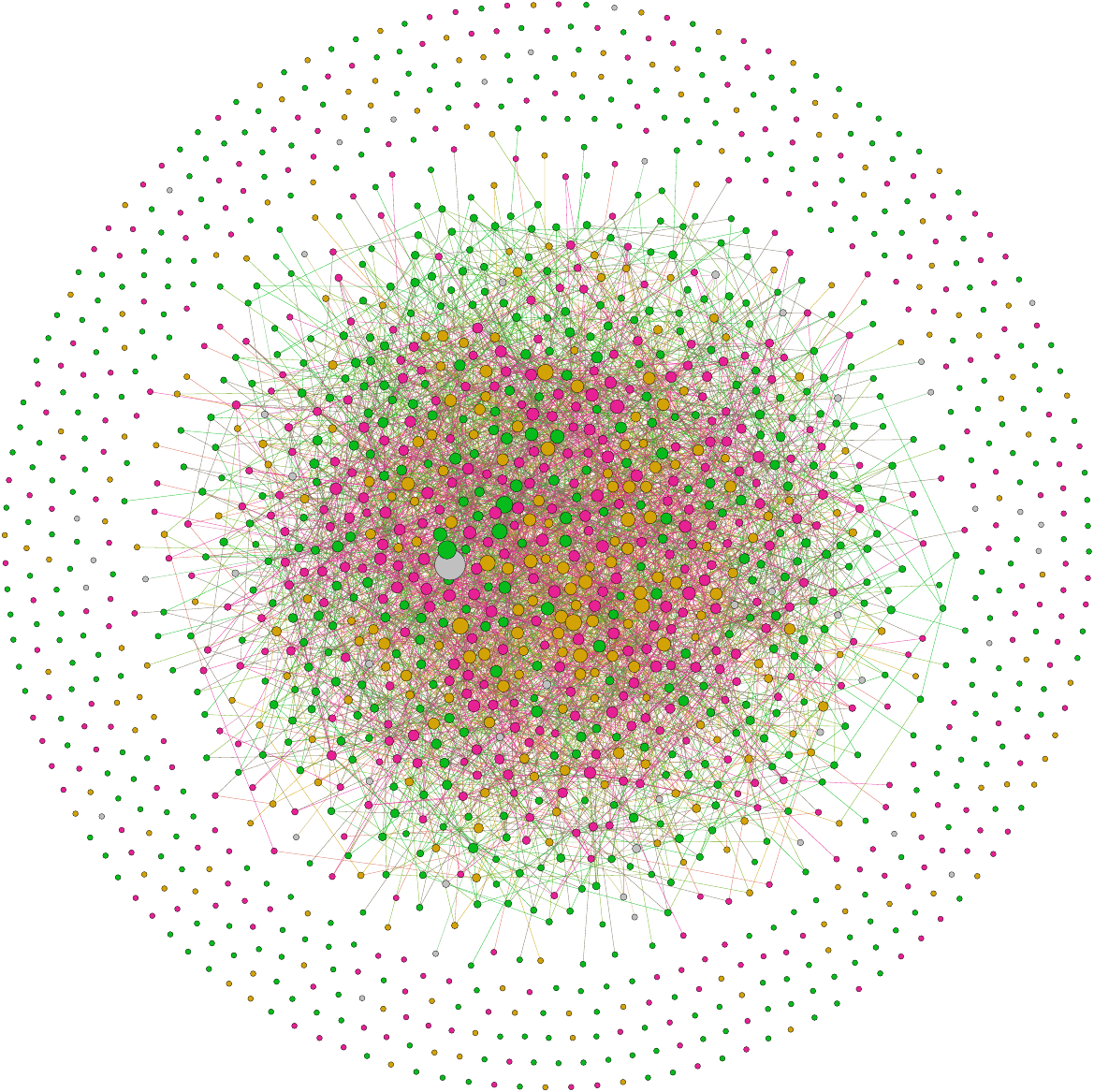}}
		\subfloat[GraphRNN]{\label{fig: pubmed_graphrnn}
			\includegraphics[width=0.19\textwidth]{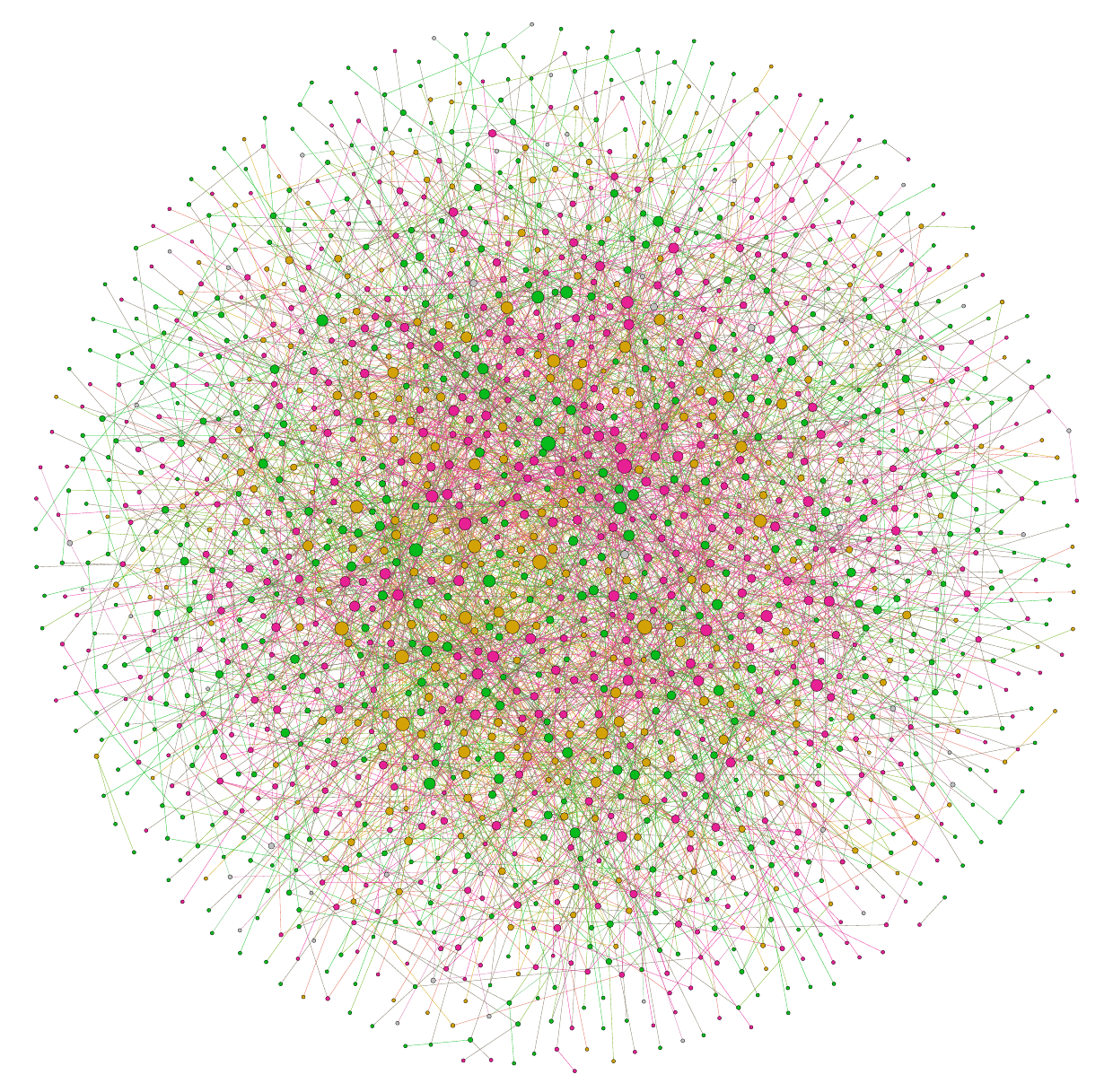}}
		\subfloat[VGAE]{\label{fig: pubmed_vgae}
			\includegraphics[width=0.19\textwidth]{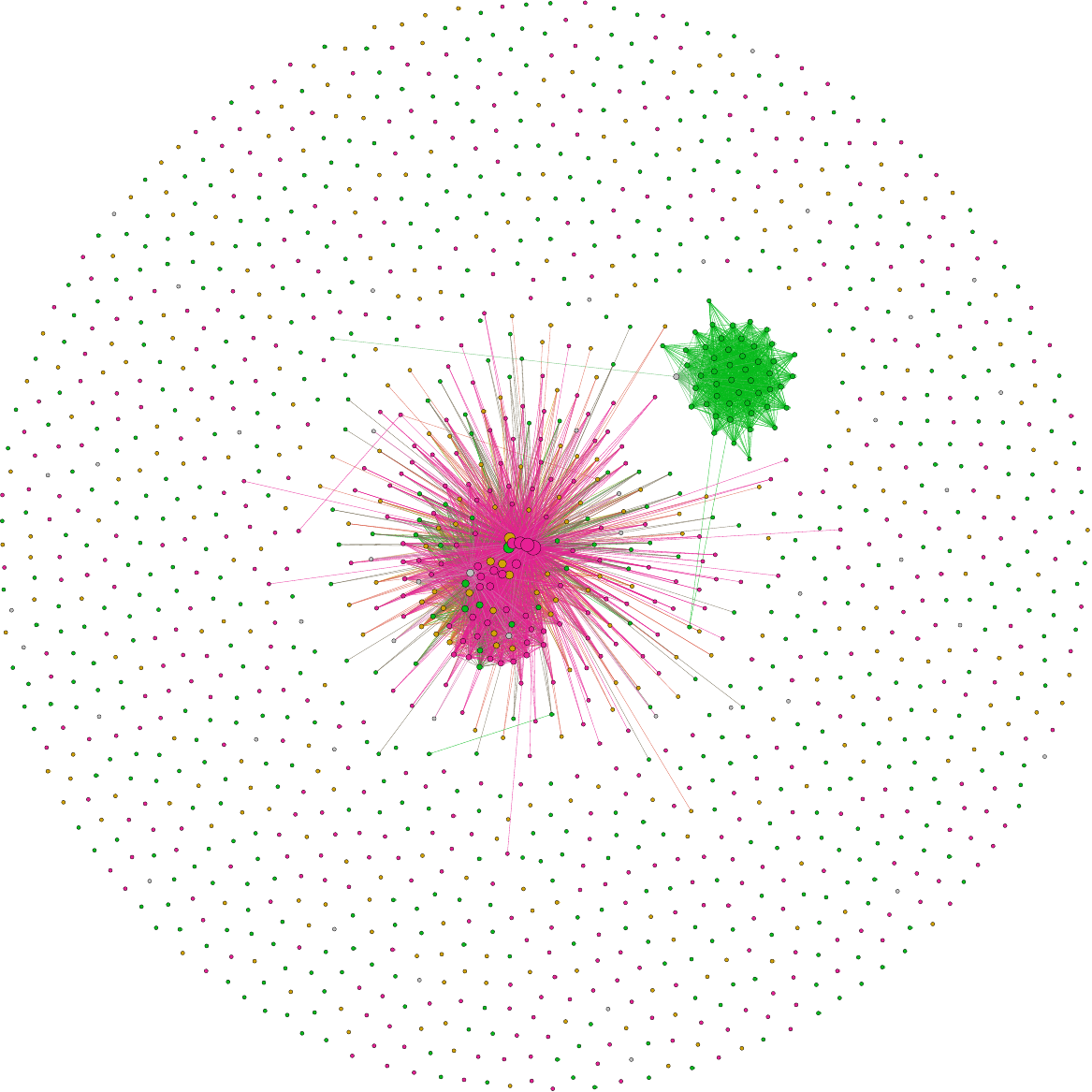}}
		\subfloat[HGEN]{\label{fig: pubmed_hgen}
			\includegraphics[width=0.19\textwidth]{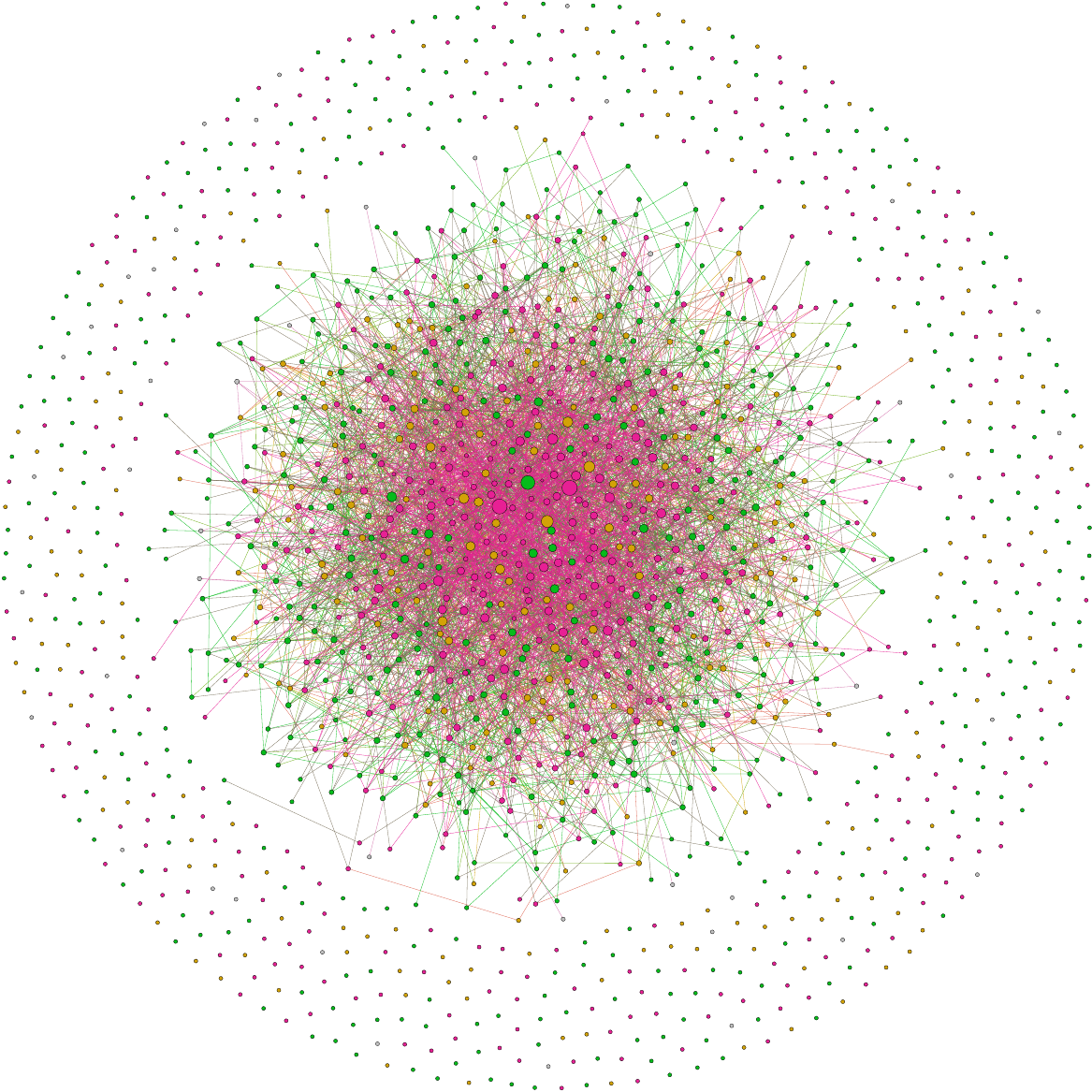}}
		\caption{\ref{fig: syn_200_real} - \ref{fig: syn_200_hgen} are the generated graph of the Syn\_200 dataset, and \ref{fig: pubmed_real} - \ref{fig: pubmed_hgen} are the generated graphs of the PubMed dataset.}
		\label{fig: syn_200}
		\vspace{-5mm}
	\end{figure*}

    \textbf{Preservation of Graph Semantic Properties}
    To further demonstrate the performance of HGEN, we evaluate the performance of meta-path distribution preservation with other baselines. Specifically, we measure the meta-path distribution from two aspects: 1) the overall meta-path length ratio preservation in generated graphs and 2) frequent meta-path patterns under each length. The results of Syn\_500 and PubMed datasets are illustrated in Fig. \ref{fig: meta-path-preservation}. In general, all the methods can approximately maintain the meta-path length ratio except for VGAE. However, HGEN can constantly achieve a better performance as shown in Fig. \ref{fig: path_ratio_syn_500} and \ref{fig: path_ratio_pubmed}. 2) As shown in Fig. \ref{fig: frequent_path_ratio_syn_500_len_2} - \ref{fig: frequent_path_ratio_syn_500_len_4} and  \ref{fig: frequent_path_ratio_pubmed_len_2} - \ref{fig: frequent_path_ratio_pubmed_len_4}, HGEN can outperform other methods by at least $10\%$ in preserving the ratio of specific meta-path patterns under each length, which is expected since HGEN is able to learn and maintain the meta-path distribution from the observed graphs while others cannot. 
    \begin{figure}[!t]
    \vspace{-5mm}
		\subfloat[Synthetic Dataset Running Time]{\label{fig: running_1}
			\includegraphics[width=0.235\textwidth]{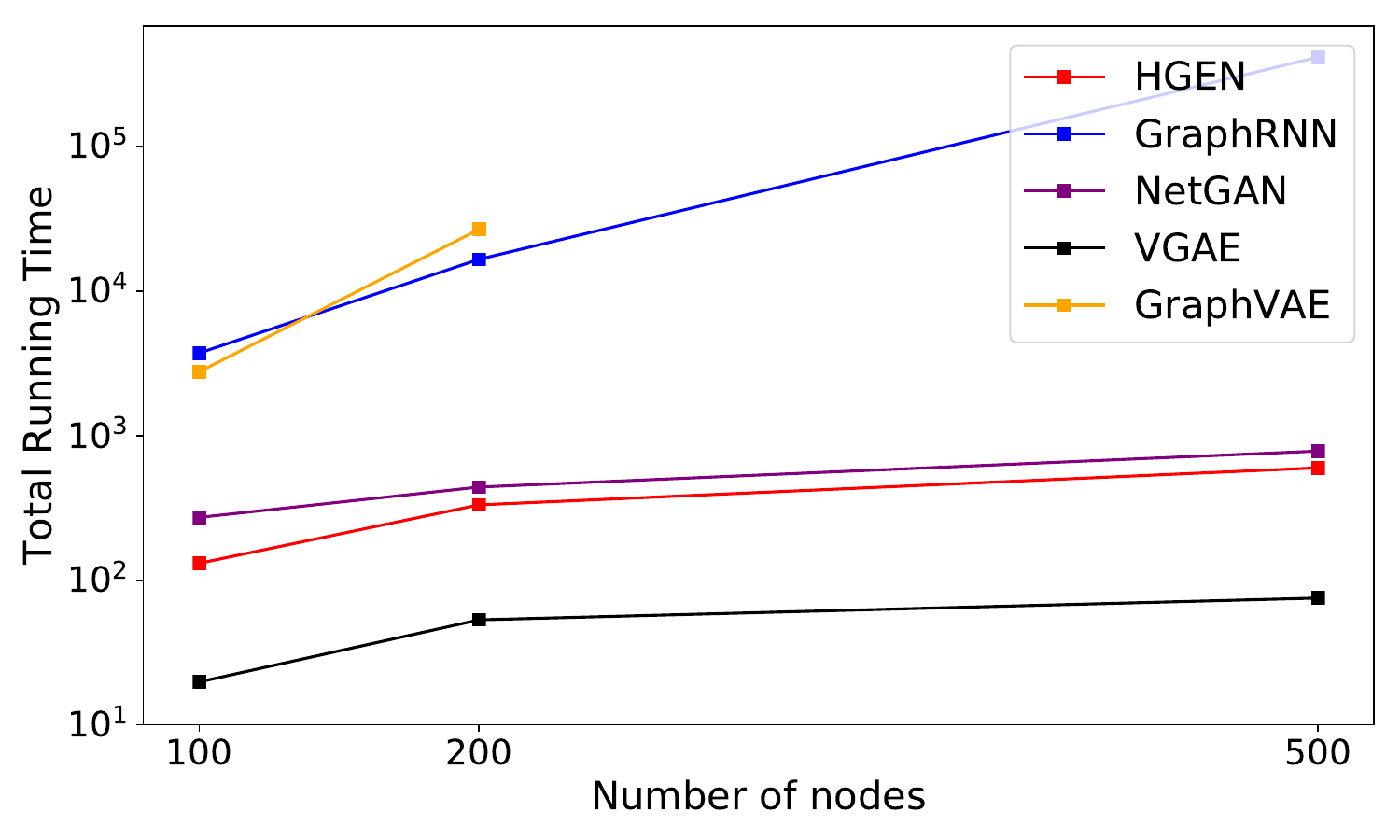}}
		\subfloat[Real Dataset Running Time]{\label{fig: running_2}
			\includegraphics[width=0.235\textwidth]{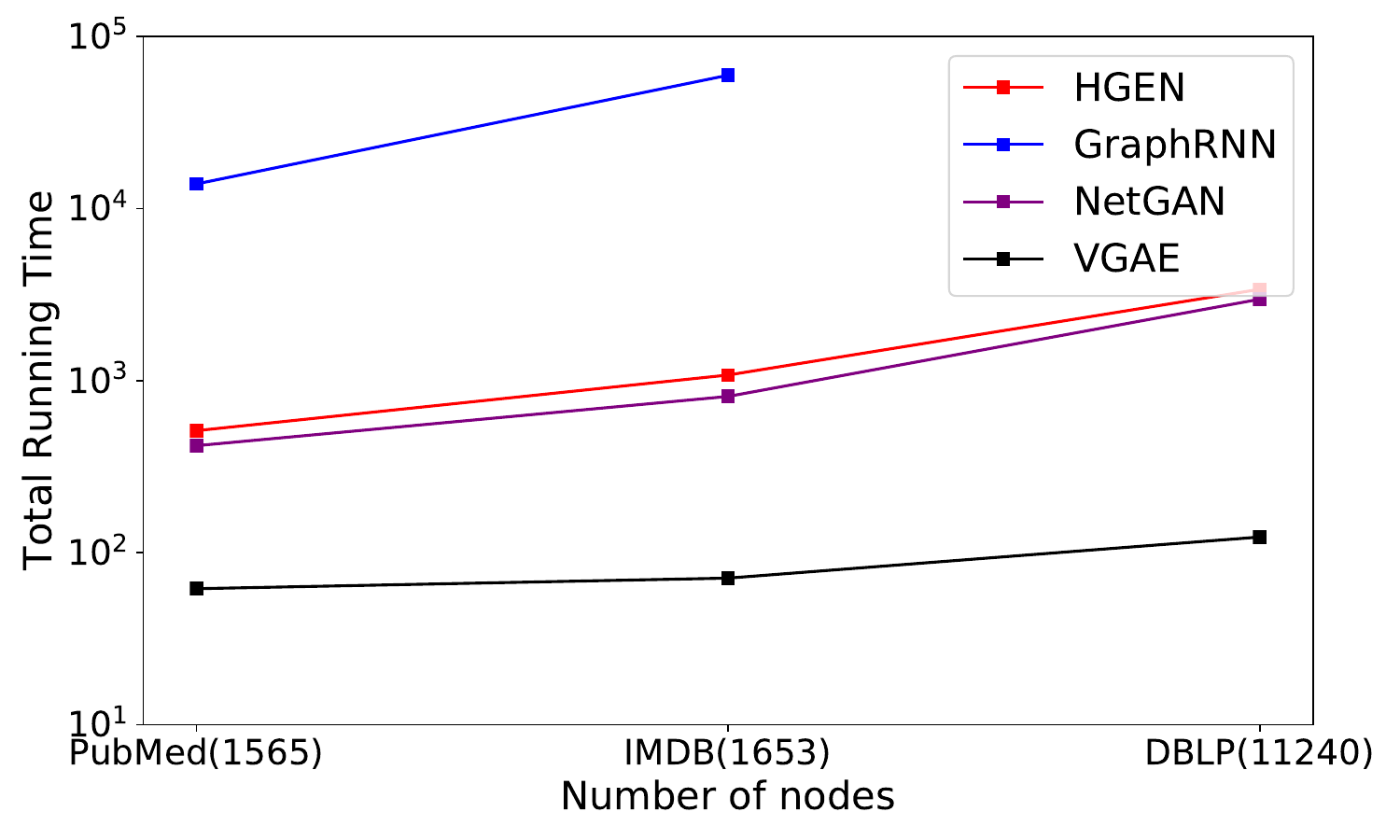}}
			\vspace{-2mm}
		\caption{Running time comparison of different models in both synthetic and real world datasets. It is clear that GraphVAE is not scalable in generating graphs with more than $200$ nodes. GraphRNN also fails in generating large graphs (with more than $10,000$ nodes). The proposed HGEN exhibits a linear running time growth in terms of the growth of graph size.}
		\label{fig: running_time}
		\vspace{-6mm}
	\end{figure}

    \subsection{Ablation Study}
    We further conduct ablation studies on the PubMed dataset to evaluate the effect of different components in HGEN, and the results are exhibited in Table \ref{tbb: ablation}. The ablative experiments are conducted based on each of the essential components in our architecture. Specifically, we select a single large heterogeneous walk length - $8$ to replace the heterogeneous walk length $1$, $2$, and $3$ in our model, and the resulting model is called HGEN-S. We also independently remove the heterogeneous node embedding to let the generator uniformly sample the next node, and the resulting model is named HGEN-E. Lastly, we replace the heterogeneous graph assembler with a probability-based graph assembler, namely HGEN-A.

    As shown in Table \ref{tbb: ablation}, all the ablative models achieve similar results in node-level metrics like Powerlaw Coef., Assortativity, which is because HGEN can well capture this node-level information through learning the heterogeneous walk distribution. Other than that, we observe: 1) HGEN-S can construct a larger sub-graph since the length of the heterogeneous walk is largely greater than HGEN, but the large subgraph doe not makes any improvements in terms of capturing the heterogeneous structural information. The reason is there are rarely long meta-path in the heterogeneous graph since longer meta-paths are highly redundant because of the shared sub-parts \cite{sun2011pathsim}. We instead choose $1$, $2$, and $3$ as our meta-path lengths to make the whole generation more flexible. 2) removing the heterogeneous node embedding would make HGEN-E hard to capture the local graph structure since HGEN relies on the encoded neighborhood information to make the node sampling be aware of the local structure. 3) as shown in the node degree distribution evaluation, replacing the heterogeneous graph assembler with a probabilistic graph assembler would cause HGEN-A hard to capture the latent heterogeneous node distribution because it uniformly samples edges from the generated walks and completely neglects the generated meta-path information. However, HGEN takes meta-paths as a basic unit to sample edges so that it can effectively preserve the overall distribution of meta-paths as proved in Theorem \ref{thm: 1}. Therefore, the node degree distribution under each type can be well preserved. 
    
    \subsection{Running Time Comparison}
    The results of our running time experiments are shown in Fig. \ref{fig: running_time}. The running times on both synthetic and real-world datasets including both training and inference time are shown with respect to the growth of number of nodes in both synthetic and real-world datasets. All running times are in $log10$ scale. As shown in both figures, random-walk-based generative models (HGEN and NetGAN) have a constant running time growth in terms of number of nodes, which is especially important when dealing with large graphs. Even though VGAE is much faster regarding running time, it is indeed a representation learning framework based on GCN and lacks of the ability of generating realistic heterogeneous graphs, and the results are also reflected in Table \ref{tab: statistical_graph_metrics}. Both GraphRNN and GraphVAE fail to compare with HGEN in model scalability because their designs require at least $O(|\mathcal{V}|^2)$ to process the transformed node sequence and adjacency matrix.
    
    \subsection{Graph Visualization}
    Since it is nearly impossible to judge whether a graph is realistic only by statistics, we visualize the generated graph to further demonstrate the performance of HGEN (Fig. \ref{fig: syn_200}). Visually, HGEN looks the most similar, while both GraphVAE and VGAE is the most dissimilar. This result is consistent with the quantitative results obtained in Table \ref{tab: statistical_graph_metrics}. For one-shot based generative models, GraphVAE and VGAE, they fail to capture the structural similarity of the observed heterogeneous graph. For the sequential-based and random walk based graph generative methods, GraphRNN and NetGAN can successfully mimic the structure similarity but fail to preserve the global heterogeneous graph properties (e.g., overall meta-path ratio).

    \section{Conclusion} \label{sec: con}
    In this paper, we propose a novel framework - HGEN for heterogeneous graph generation, which can jointly capture the semantic, structural, and global distributions of heterogeneous graphs. Our framework consists of a novel heterogeneous walk generator that can hierarchically generate meta-path instances (namely heterogeneous walk) and a heterogeneous graph assembler that can construct new graphs by sampling from the generated heterogeneous walks in a stratified manner. Extensive experiments on synthetic and real-world datasets demonstrate the advantages of HGEN over existing deep generative models in terms of preserving both graph statistical and heterogeneous specified properties.
    
\bibliographystyle{IEEEtran}
\bibliography{cling}

\end{document}